\documentclass[final,leqno,onefignum,onetabnum]{siamltex1213}
\pdfoutput=1

\usepackage{graphicx}
\usepackage{subfigure}
\usepackage{amsmath}
\usepackage{enumerate}
\usepackage{amssymb}
\usepackage{verbatim}
\usepackage[colorinlistoftodos]{todonotes}
\usepackage{array}

\usepackage{todonotes}
\usepackage{booktabs}
\usepackage{siunitx}
\usepackage{braket}
\usepackage{bm} 
\usepackage{bbm} 
\usepackage[ruled]{algorithm2e}

\newcommand{\ind}{\operatorname{\mathbbm{1}}}

\newcommand{\Phidet}{\Phi_\textrm{d}}

\DeclareMathOperator{\prob}{Pr}
\DeclareMathOperator{\round}{round}

\renewcommand{\d}{\mathrm{d}}

\newcommand{\XSet}{\mathbb{X}}

\newcommand{\USet}{\mathbb{U}}
\newcommand{\RSet}{\mathbb{R}}
\newcommand{\NSet}{\mathbb{N}}
\newcommand{\ZSet}{\mathbb{Z}}
\newcommand{\SSet}{\mathbb{S}}
\newcommand{\SnSet}{\mathbb{S}_N}

\newcommand{\CSet}{\mathbb{C}}
\newcommand\blank{{\mkern 2mu\cdot\mkern 2mu}}
\newcommand{\eps}{\varepsilon}

\graphicspath{ {Figures/} }

\begin{document}

\title{Macroscopic coherent structures in a stochastic neural network: from
interface dynamics to coarse-grained bifurcation analysis.} 
\author{ Daniele Avitabile
\thanks{Centre for Mathematical Medicine and Biology, School of Mathematical Sciences, University of Nottingham, Nottingham, NG2 7RD, UK }
\and
Kyle Wedgwood
\thanks{Centre for Biomedical Modelling and Analysis, University of Exeter, RILD
  Building, Barrack Road,
Exeter, EX2 5DW, UK}
}

\maketitle

\begin{abstract}
  We study coarse pattern formation in a cellular automaton modelling a
  spatially-extended stochastic neural network. The model, originally proposed
  by Gong and Robinson \cite{Gong2012aa}, is known to support stationary and
  travelling bumps of localised activity. We pose the model on a ring and study
  the existence and stability of these patterns in various limits using a
  combination of analytical and numerical techniques. In a purely deterministic
  version of the model, posed on a continuum, we construct bumps and travelling waves
  analytically using standard interface methods from neural fields theory. In a
  stochastic version with Heaviside firing rate, we construct approximate
  analytical
  probability mass functions associated with bumps and travelling waves. In the
  full stochastic model posed on a discrete lattice, where a coarse analytic
  description is unavailable, we compute patterns and their linear stability
  using equation-free methods. The lifting procedure used in the coarse
  time-stepper is informed by the analysis in the deterministic and stochastic
  limits. In all settings, we identify the synaptic profile as a mesoscopic
  variable, and the width of the corresponding activity set as a macroscopic
  variable. Stationary and travelling bumps have similar meso- and macroscopic
  profiles, but different microscopic structure, hence we propose lifting
  operators which use microscopic motifs to disambiguate between them.
  We provide numerical evidence that waves are supported by a
  combination of high synaptic gain and long refractory times, while
  meandering bumps are elicited by short refractory times.
\end{abstract}

\section{Introduction}
In the past decades, single-neuron recordings have been complemented by
multineuronal experimental techniques, which have provided quantitative evidence
that the cells forming the nervous systems are coupled both
structurally~\cite{Braitenberg:1998iw} and functionally (for a recent review,
see~\cite{yuste2015aa} and references therein). An important question in
neuroscience concerns the relationship between electrical activity at the level
of individual neurons and the emerging spatio-temporal coherent structures
observed experimentally using local field potential
recordings~\cite{Einevoll:2013hp}, functional magnetic resonance
imaging~\cite{vandenHeuvel:2010fz} and
electroencephalography~\cite{Nunez:2006ir}. 

There exist a wide variety of models describing activity at the level of an
individual neuron~\cite{Izhikevich:2007vr,Ermentrout:2010cg}, and major research
efforts in theoretical and computational neuroscience are directed towards coupling
neurons in large-dimensional neural networks, whose behaviour is studied mainly via
direct numerical simulations~\cite{Izhikevich2008aa,Fairhall:2014kb}.

A complementary approach, dating back to Wilson and
Cowan~\cite{Wilson1972aa,Wilson1973aa} and Amari~\cite{Amari1975aa,Amari1977aa}, 
foregoes describing activity at the single neuron level by representing averaged
activity across populations of neurons. These \textit{neural field
models} are nonlocal, spatially-extended, excitable pattern-forming
systems~\cite{Ermentrout:1998ct} which are often analytically tractable and
support several coherent structures such as localised radially-symmetric
states~\cite{Werner2001aa,Laing2002aa,Laing2003aa,Bressloff2011aa,Faye2013aa},
localised patches~\cite{Laing2003q,Rankin:2014bz,Avitabile2015aa},
patterns on lattices with various
symmetries~\cite{Ermentrout:1979dp,Bressloff:2001gt},
travelling bumps and fronts~\cite{Ermentrout:1993kc,Bressloff:2014cm},
rings~\cite{Owen2007aa,Coombes2012aa},
breathers~\cite{Folias2004aa,Folias:2005hq,Folias:2012ez}, 
target patterns~\cite{Coombes2013aa}, 
spiral waves~\cite{Laing2005aa} 
and lurching waves~\cite{Golomb:1999cr,Osan:2001vb,Wasylenko:2010je}
(for comprehensive reviews, we refer the reader
to~\cite{Bressloff2012o,Bressloff:2014cm}).

Recent studies have analysed neural fields with additive
noise~\cite{Hutt2008aa,Faugeras2009aa,Kuehn2014aa},
multiplicative noise~\cite{Bressloff2012ab}, or noisy firing
thresholds~\cite{Brackley2007aa}, albeit these models are still mostly
phenomenological. Even though several papers derive continuum neural fields from
microscopic models of coupled
neurons~\cite{Jirsa:1997kq,Bressloff2009aa,Bressloff:2010jc,Baladron:2012fs},
the development of a rigorous theory of multi-scale brain models is an active
area of research.

Numerical studies of networks based on realistic neural biophysical models rely
almost entirely on brute-force Monte Carlo simulations (for a very recent, remarkable
example, we refer the reader to~\cite{Markram2015}).
With this \textit{direct
numerical simulation}
approach, the stochastic evolution of each neuron in the network is monitored,
resulting in huge computational costs, both in terms of computing time and memory.
From this point of view, multi-scale numerical techniques for neural networks present
interesting open problems.

When few clusters of neurons with similar properties form in the network,
a significant reduction in computational costs can be obtained by population density
methods~\cite{Omurtag:2000dq,Haskell:2001bc}, which evolve probability density
functions of neural subpopulations, as opposed to single neuron trajectories. This
coarse-graining technique is particularly effective when the underlying microscopic
neuronal model has a low-dimensional state space (such as the leaky
integrate-and-fire model) but its performance degrades for more realistic biophysical
models. Developments of the population density method involve analytically derived
moment closure approximations~\cite{Cai:2004em,Ly:2007fg}. Both Monte Carlo
simulations and population density methods give access only to stable asymptotic
states, which may form only after long-transient simulations. 

An alternative approach is offered by
\textit{equation-free}~\cite{Kevrekidis2003r,Kevrekidis:2009jo} and
\textit{heterogeneous multiscale} methods~\cite{Weinan2003aa,Weinan:2007tra}, which implement
multiple-scale simulations using an on-the-fly numerical closure
approximations. Equation-free methods, in particular, are of interest in
computational neuroscience as they accelerate macroscopic simulations and allow
the computation of unstable macroscopic states. In addition, with equation-free
methods, it is possible to perform coarse-grained bifurcation analysis using
standard numerical bifurcation techniques for time-steppers~\cite{Tuckerman:2000jn}.

The equation-free framework~\cite{Kevrekidis2003r,Kevrekidis:2009jo} assumes the
existence of a closed coarse model in terms of a few macroscopic state
variables. The model closure is enforced numerically, rather than analytically, using 
a \textit{coarse time-stepper}: a computational procedure which takes advantage
of knowledge of the microscopic dynamics to time-step an approximated
macroscopic evolution equation. A single coarse time step from time $t_0$ to time $t_1$ is
composed of three stages: (i) \textit{lifting}, that is, the creation of microscopic
initial conditions that are compatible with the macroscopic states at time $t_0$;
(ii) \textit{evolution}, the use of independent realisations of the microscopic model
over a time interval $[t_0, t_1]$; (iii) \textit{restriction}, that is, the
estimation of the macroscopic state at time $t_1$ using the realisations of the
microscopic model.

While equation-free methods have been employed in various contexts
(see~\cite{Kevrekidis:2009jo} and references therein) and in particular in
neuroscience
applications~\cite{Laing:2006bg,Laing:2007bna,Laing:2010fw,Spiliotis2011aeta,Spiliotis:2012bz,
Laing:2015gp}, there are still open questions, mainly related to how noise
propagates
through the coarse time stepper. A key aspect of every equation-free
implementation is the lifting step. The underlying lifting operator, which maps a
macroscopic state to a set of microscopic states, is generally non-unique, and
lifting choices have a considerable impact on the convergence properties of the
resulting numerical scheme~\cite{Avitabile2014aa}. Even though the choice of coarse
variables can be automatised using data-mining techniques, as shown in several papers
by Laing, Kevrekidis and co-workers~\cite{Laing:2006bg,Laing:2007bna,Laing:2010fw},
the lifting step is inherently problem dependent.

The present paper explores the possibility of using techniques from neural field
models to inform the coarse-grained bifurcation analysis of discrete neural
networks. A successful strategy in analysing neural fields is to replace the
models' sigmoidal firing rate functions with Heaviside
distributions~\cite{Bressloff2012o,Bressloff:2014cm}. Using this strategy, it is
possible to relate macroscopic observables, such as bump widths or wave speeds,
to biophysical parameters, such as firing rate thresholds. Under this
hypothesis, a macroscopic variable suggests itself, as the state of the system
can be constructed entirely via the loci of points in physical space where the
neural activity attains the firing-rate threshold value. In addition, there
exists a closed (albeit implicit) evolution equation for such
interfaces~\cite{Coombes2012aa}.

In this study, we show how the insight gained in the Heaviside limit may be
used to perform coarse-grained bifurcation analysis of neural networks, even in cases
where the network does not evolve according to an integro-differential equation. As
an illustrative example, we consider a spatially-extended neural
network in the form of a discrete time Markov chain with discrete ternary state space,
posed on a lattice.
The model is an existing cellular automaton proposed by Gong and
co-workers~\cite{Gong2012aa}, and it has been related to neuroscience in the context
of relevant spatio-temporal activity patterns that are observed in cortical tissue.
In spite of its simplicity, the model possesses sufficient complexity to support rich
dynamical behaviour akin to that produced by neural fields. In particular, it
explicitly includes refractoriness and is one of the simplest models capable of
generating propagating activity in the form of travelling waves. An important feature
of this model is that the microscopic transition probabilities depend on the local
properties of the tissue, as well as on the global synaptic profile across the network.
The latter has a convolution structure typical of neural field models, which we
exploit to use interface dynamics and define a suitable lifting strategy.

To connect our micro- and macroscopic variables, we take advantage of
interface approaches, which are typically applied to continuum networks. A notable
exception is offered by Chow and Coombes~\cite{Chow2006aa}, who consider a network
based upon the lighthouse model~\cite{Haken2002}. In a similar vein to our
approach, they show how analysis of the discrete network can be facilitated by
considering a continuum approximation and derive threshold equations to define bump
solutions. This analysis also highlights that perturbations to the microscopic
state, specifically the phase arrangement within the bump, can alter the dynamics
of the bump edges.

Chow and Coombes found that wandering bump solutions in the lighthouse model
arise for sufficiently fast synaptic processing. This is
congruent with our result that short refractory times in~\eqref{eq:microPhi}
elicit coherent bump states, since both refractory times and synaptic processing
timescales affect the average firing rate of the neuron.
However, bumps cease to exist in our model if the refractory times are too long,
whereas the lighthouse model supports stationary bumps for slow synapses, which
highlights the subtle differences between the roles of refractoriness and synaptic
processing in neural networks.
It should also be noted that the meandering observed, for instance, in
Figure~\ref{fig:exampleBump} is due to noise, and that all bumps will tend to
wander; on the other hand, the meandering described by Chow and Coombes arises from
the deterministic dynamics of the lighthouse model, and it is triggered by a
sufficently fast synaptic process. We also remark that, without
modification, the lighthouse model does not support travelling wave solutions, and
so we cannot make comparisons regarding these solutions.

We initially study the model in simplifying limits in which an analytical (or
semi-analytical) treatment is possible. In these cases, we construct bump and wave
solutions and compute their stability. This analysis follows the standard Amari
framework, but is here applied directly to the cellular automaton.
We then derive the corresponding lifting operators, which highlight a critical
importance of the microscopic structure of solutions: one of the main results of our
analysis is that, since macroscopic stationary and travelling bumps coexist and have
nearly identical macroscopic profiles, a standard lifting is unable to distinguish
between them, thereby preventing coarse numerical continuation. These
structures, however, possess different microstructures, which are captured
by our analysis and subsequently by our lifting operators. 
This allows us to compute separate solution branches, in which we vary several model
parameters, including those associated with the noise processes.

The manuscript is arranged as follows: In Section~\ref{sec:modelDescription} we
outline the model. In Section
\ref{sec:simulations}, we simulate the model and identify the macroscopic
profiles in which we are interested, together with the coarse variables that
describe them. In Section~\ref{sec:macroModel}, we define a deterministic
version of the full model and lay down the framework for analysing it. In
Sections~\ref{sec:MacroBump} and
\ref{sec:TWDet}, we respectively construct bump and wave solutions under
the deterministic approximation and compute the stability of these solutions.
In Section~\ref{sec:approxMu}, we define and construct travelling waves
relaxing the deterministic limit. In Sections~\ref{sec:bumpLift} and
\ref{sec:TWpLift}, we provide the lifting steps for use in the equation-free
algorithm for the bump and wave respectively. In Section~\ref{sec:continuation},
we briefly outline the continuation algorithm and in Section~\ref{sec:numerics},
we show the results of applying this continuation to our system. Finally, in
Section
\ref{sec:discussion}, we make some concluding remarks.

\section{Model description} \label{sec:modelDescription}
\subsection{State variables for continuum and discrete tissues} \label{sec:genModel}
In this section, we present a modification of a model originally proposed by Gong and
Robinson~\cite{Gong2012aa}.
We consider a one-dimensional
neural tissue 
$\XSet \subset \RSet$. 
At each discrete time step
$t \in \ZSet$, a neuron at position $x \in \XSet$ may be in one of three states: a
refractory state (henceforth denoted as $-1$), a quiescent state (0) or a spiking
state (1). Our state variable is thus
a function $u \colon \XSet \times \ZSet \to \USet$, where
$\USet=\set{-1,0,1}$.  
%
We pose the model on a continuum tissue $\SSet = \RSet/2L\ZSet$ or on a discrete
tissue featuring $N+1$ evenly spaced neurons,
\[
\SSet_N = \{x_i\}_{i=0}^N, \qquad x_i = -L + i2L /N \in [-L,L].
\]
We will often alternate between the discrete and the continuum setting, hence we will
use a unified notation for these cases. We use the symbol $\XSet$ to refer to either
$\SSet$ or $\SnSet$, depending on the context. Also, we use $u(\blank,t)$ to
indicate the state variable in both the discrete and the continuum case:
$u(\blank,t)$ will denote a step function defined on $\SSet$ in the continuum
case and a vector in $\USet^N$ with components $u(x_i,t)$ in the discrete case.
Similarly, we write $\int_\XSet u(x) \,\d x$ to indicate $\int_\SSet u(x) \, \d
x $ or $2L/N \sum_{j=0}^N
u(x_j)$.

\subsection{Model definition} \label{sec:microModel}
We use the term \textit{stochastic model} when the Markov chain model
described below is posed on $\SnSet$.
An example of a state supported by the stochastic model is given in
Figure~\ref{fig:schematic}(a).
\begin{figure}
  \centering
  \includegraphics{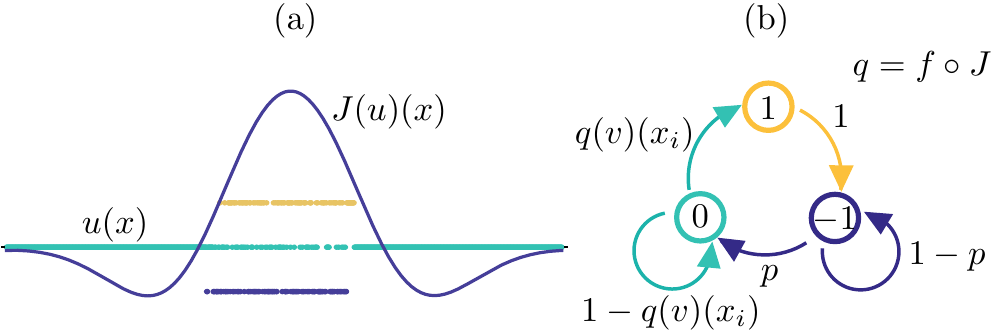}
  \caption{(a): Example of a state $u(x) \in \USet^N$ and corresponding
  synaptic profile $J(u)(x) \in \RSet^N$ in a stochastic network of $1024$ neurons. (b):
  Schematic of the transition kernel for the network (see also
  Equations~\eqref{eq:probRefr}--\eqref{eq:probSpike}). The conditional
  probability of the \textit{local} variable $u(x_i,t+1)$ depends on the
  \textit{global} state of the network at time $t$, via the function $q = f
  \circ J$, as seen in~\eqref{eq:probSpike}.}
  \label{fig:schematic}
\end{figure}

In the model, neurons are coupled via a translation-invariant synaptic kernel, that is, we
assume the connection strength between two neurons to be
dependent on their mutual Euclidean distance. In particular, we prescribe that short
range connections are excitatory, whilst long-range connections are inhibitory. To
model this
coupling, we use a standard Mexican hat function, 
\begin{equation} \label{eq:kernel}
   w \colon \XSet \to \RSet, \qquad
   x \mapsto A_1 \sqrt{B_1/L} \exp(-4 B_1 x^2) - A_2 \sqrt{B_2/L} \exp(-4 B_2
   x^2),
\end{equation}
and denote by $W$ its periodic extension.

In order to describe the dynamics of the model, it is useful to partition the
tissue $\XSet$ into the $3$ pullback sets
\begin{equation} \label{eq:partitions}
  X^u_k(t) = \Set{ x \in \XSet \colon u(x,t) = k }, \quad k \in \USet, \quad t
  \in \ZSet,
\end{equation}
so that we can write, for instance, $X^u_1(t)$ to denote the set
of neurons that are firing at time $t$ (and similarly for $X^u_{-1}$ and $X^u_0$).
Where it is unambiguous, we shall simply write $X_k$ or $X_k(t)$ in place of $X^u_k(t)$.

The synaptic input to a cell at position $x_i$ is given by a weighted sum of
inputs from all firing cells. Using the synaptic kernel~\eqref{eq:kernel} and
the partition~\eqref{eq:partitions}, the synaptic input 
is then modelled as
\begin{equation} \label{eq:J}
    J \colon \XSet \times \ZSet \to \RSet, 
    \qquad
    (x,t) \mapsto
	   \kappa \int_{\XSet} W(x - y) \ind_{X_1(t)}(y) \, \d y
         = \kappa \int_{X_1(t)} W(x - y) \, dy,
\end{equation}
where $\kappa \in \RSet_+$ is the synaptic gain, which is common for all
neurons
and $\ind_X$ is the indicator function of a set $X$. 

\begin{rem}[Synaptic input as mesoscopic variable] \label{rem:uDependence}
  Since $X_1$ depends on the microscopic state variable $u$, so does the
  synaptic input~\eqref{eq:J}. 
  Where necessary, we will write $J(u)(x,t)$ to highlight the dependence on $u$. We
  refer the reader to
  Figure~\ref{fig:schematic} for a concrete example of synaptic profile. 
\end{rem}

The firing probability associated to a quiescent neuron is linked to the
synaptic input via the firing rate function
\begin{equation} \label{eq:firingRate}
  f \colon \RSet \to \RSet,
  \qquad
  I \mapsto \frac{1}{1+ \exp[-\beta(I - h)]},
\end{equation}
whose steepness and threshold are denoted by the positive real numbers $\beta$
and $h$, respectively.
We are now ready to describe the evolution of the stochastic model, which 
is a discrete-time Markov process with finite state space $\USet^N$ and
transition probabilities specified as follows: for each $x_i \in \SnSet$ and $t \in
\ZSet$
\begin{align}
  & \prob \big[ u(x_i,t+1) = -1 \big| u(x,t) = v(x) \big] 
  =
  \begin{cases}
    1-p & \text{if $v(x_i)=-1$,} \\
      1 & \text{if $v(x_i)= 1$,} \\
      0 & \text{otherwise,}
  \end{cases}
  \label{eq:probRefr}
  \\
  & \prob \big[ u(x_i,t+1) = 0 \big| u(x,t) = v(x) \big] 
  =
  \begin{cases}
      p            & \text{if $v(x_i)=-1$,} \\
    1-f(J(v))(x_i) & \text{if $v(x_i)= 0$,} \\
      0 & \text{otherwise,}
  \end{cases}
  \label{eq:probQuies}
  \\
  & \prob \big[ u(x_i,t+1) = 1 \big| u(x,t) = v(x) \big] 
  =
  \begin{cases}
    f(J(v))(x_i) & \text{if $v(x_i)= 0$,} \\
    0 & \text{otherwise,}
  \end{cases}
  \label{eq:probSpike}
\end{align}
where $p \in (0,1]$. We give a schematic representation of the transitions of
each neuron in the network
in Figure~\ref{fig:schematic}(b). We remark that conditional probability of the
\textit{local} variable $u(x_i,t+1)$ depends on the \textit{global} state of the
network at time $t$, via the function $f \circ J$.

The model described by~\eqref{eq:kernel}--\eqref{eq:probSpike}, complemented by
initial conditions, defines a stochastic evolution map that we will formally denote as
\begin{equation}\label{eq:microPhi}
  u(x,t+1) = \varphi(u(x,t); \gamma),
\end{equation}
where $\gamma=(\kappa,\beta,h,p,A_1,A_2,B_1,B_2)$ is a vector of control parameters.  

\begin{rem}[Microscopic, mesoscopic and macroscopic descriptions] \label{rem:microMesoMacro}
  We will henceforth use the terms ``microscopic'', ``mesoscopic'' and
  ``macroscopic'' to refer to different state variables or model descriptions.
  Examples of these three state variables appear toghether in
  Figures~\ref{fig:exampleBump}--\ref{fig:exampleWave} in
  Section~\ref{sec:simulations}, and we introduce them briefly here:
  \begin{description}
    \item[Microscopic level.] Model~\eqref{eq:microPhi} will be
      referred to as microscopic model and its solutions at a fixed time $t$ as
      microscopic states. We will use these terms also when $p=1$ and $\beta \to
      \infty$, that is, when the evolution equation~\eqref{eq:microPhi} is
      deterministic.
    \item[Mesoscopic level.] In Remark~\ref{rem:uDependence}, we associated to
      each microscopic state $u$ a corresponding synaptic profile $J$, which is
      smooth, even when the tissue is discrete. We will not seek for an evolution
      equation for the variable $J$, as the corresponding dynamical system would not
      reprent a reduction of the microscopic one. However, we will use $J$ to bridge
      between the microscopic and macroscopic model descriptions; we therefore refer
      to $J$ as a mesoscopic variable (or mesoscopic state).
    \item[Macroscopic level.] Much of the present paper aims to show
      that, for the model under consideration, there exists a high-level model
      description, in the spirit of interfacial dynamics for neural
      fields~\cite{Bressloff2012o,Coombes2012aa,Bressloff:2014cm}. The state
      variables for this level are points on the tissue where $J(u)(x,t)$ attains the
      firing rate threshold $h$. We will denote these threshold crossings as
      $\{\xi_i(t)\}$ and we will discuss (reduced) evolution equations in terms of
      $\xi_i(t)$. The variables $\{ \xi(t) \}$ are therefore referred to as
      macroscopic variables and the corresponding evolution equations as macroscopic
      model. 
  \end{description}
\end{rem}

\section{Microscopic states observed via direct simulation} \label{sec:simulations}
\begin{table}
  \centering
  \begin{tabular}{l c c c  c  c c  c c c c}
    \toprule
    Experiment & {$\kappa$} & {$\beta$} & {$h$} & {$p$} & {$A_1$} & {$A_2$} & {$B_1$} & {$B_2$} & {$N$} & {$L$} \\
    \midrule
          Bump            &   30   &  5         &  0.9  & 0.7  &  5.25  &  5   & 0.2 & 0.3 & 1024 & $\pi$ \\
          Multiple bump   &   60   &  5         &  0.9  & 0.7  &  5.25  &  5   & 0.2 & 0.3 & 2058 & $2\pi$ \\
          Travelling wave &   30   &  $\infty$  &  1.0  & 0.4  &  5.25  &  5   & 0.2 & 0.3 & 1024 & $\pi$ \\
    \bottomrule
  \end{tabular}
  \caption{Parameter values for which the stochastic model supports a bump
  (Figure~\ref{fig:exampleBump}), a multiple-bump solution
  (Figure~\ref{fig:exampleMultiBump}) and a travelling wave
  (Figure~\ref{fig:exampleWave}). The value $\infty$ for the parameter $\beta$
  indicates that a Heaviside firing rate has been used in place of the sigmoidal
  function~\eqref{eq:firingRate}.}
  \label{tab:params}
\end{table}

In this section, we introduce a few coherent states supported by the stochastic model.
The main aim of the section is to show examples of bumps, multiple bumps
and travelling waves, whose existence and stability properties will be studied in
the following sections. In addition, we give a first characterisation of the
macroscopic variables of interest and link them to the microscopic structure observed
numerically.

\subsection{Bumps} \label{subsec:bumps}
In a suitable region of parameter space, the microscopic model supports bump
solutions~\cite{Qi2015aa} in which the microscopic variable $u(x,t)$ is nonzero only
in a localised region of the tissue. In this \emph{active} region, neurons attain
all values in $\USet$. In Figure~\ref{fig:exampleBump}, we
show a time simulation of the microscopic model with $N=1024$ neurons. At each time
$t$, neurons are in the refractory (blue), quiescent (green) or spiking (yellow)
state. We prescribe the initial condition by setting $u(x_i,0)=0$ outside of a
localised region, in which $u(x_i,0)$ are sampled randomly from $\USet$. After a short
transient, a stochastic \textit{microscopic} bump is formed. As
expected due to the stochastic nature of the system \cite{Kilpatrick2013}, the
active region wanders while remaining localised. A space-time section of the
active region reveals a characteristic random microstructure (see
Figure~\ref{fig:exampleBump}(a)).
\begin{figure}
  \centering
  \includegraphics{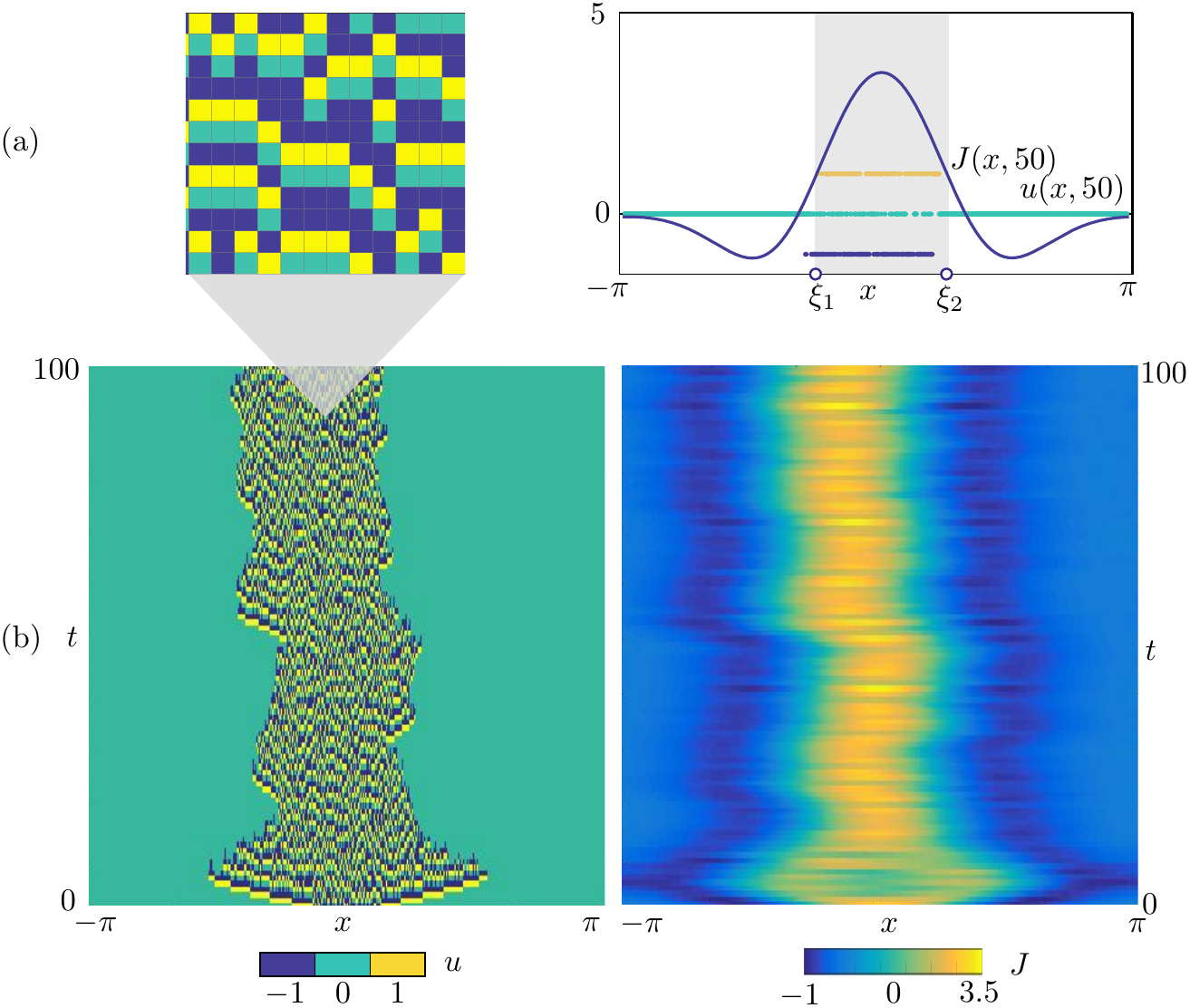}
  \caption{Bump obtained via time simulation of the stochastic model for $(x,t) \in
  [-\pi,\pi] \times [0,100]$. (a): The microscopic state $u(x,t)$
  (left) attains the discrete values $-1$ (blue), $0$ (green) and $1$ (yellow). The
  corresponding synaptic profile $J(x,t)$ is a continuous function. A comparison
  between $J(x,50)$ and $u(x,50)$ is reported on the right panel, where we also mark
  the interval $[\xi_1, \xi_2]$ where $J$ is above the firing threshold $h$. (b):
  Space-time plots of $u$ and $J$. Parameters as in Table~\ref{tab:params}.}
  \label{fig:exampleBump}
\end{figure}
By plotting $J(x,t)$, we see that the active region is well approximated by the
portion of the tissue $X_\geq = [\xi_1, \xi_2]$ where $J$ lies above the threshold
$h$. A quantitative comparison between $J(x,50)$ and $u(x,50)$ is made in
Figure~\ref{fig:exampleBump}(a). We interpret $J$ as a \textit{mesoscopic} variable
associated with the bump, and $\xi_1$ and $\xi_2$ as corresponding
\textit{macroscopic} variables (see also Remark~\ref{rem:microMesoMacro}).

\begin{figure}
  \centering
  \includegraphics{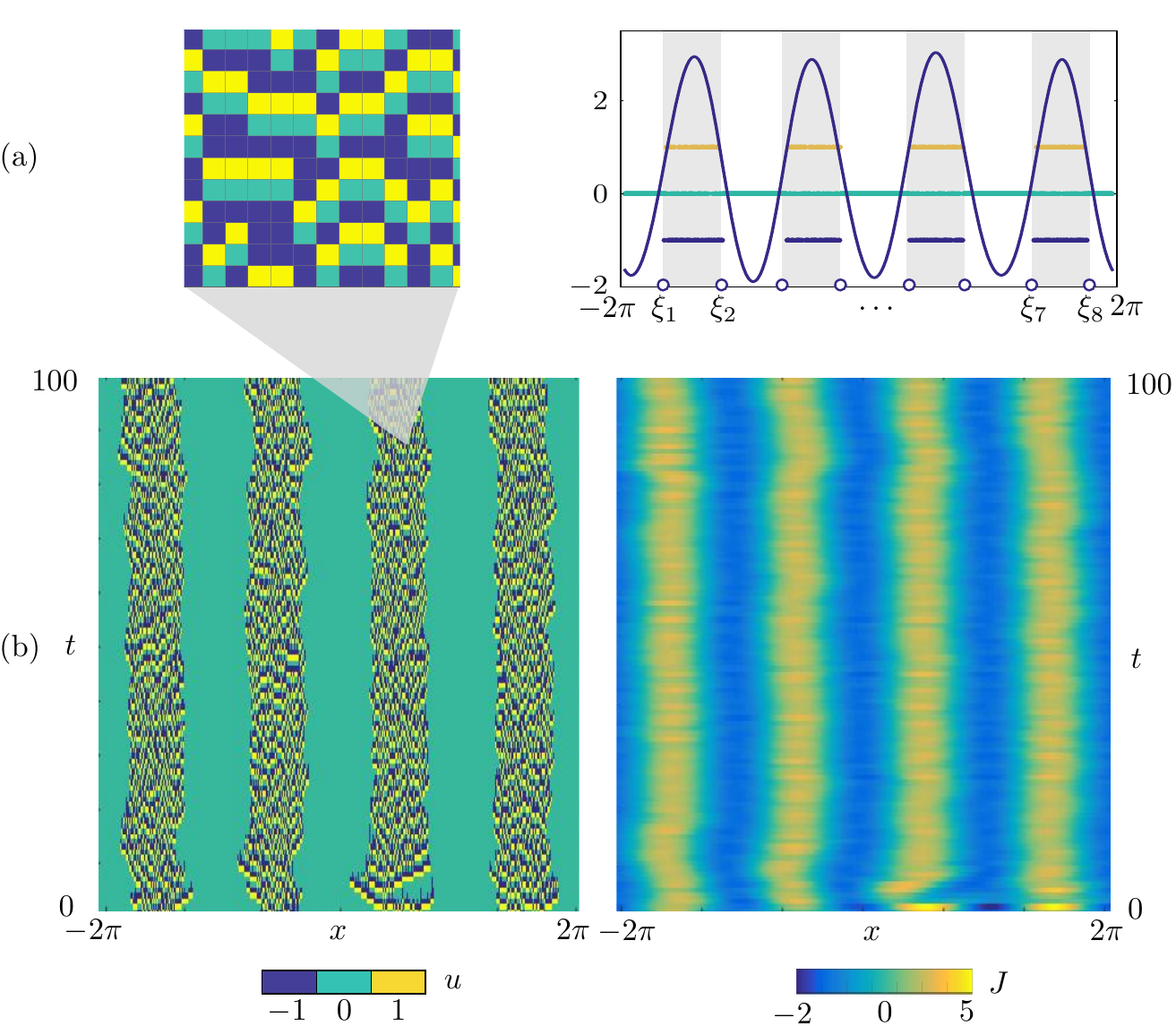}
  \caption{Multiple bump solution obtained via time simulation of the stochastic model
  for $(x,t) \in [-2\pi,2\pi] \times [0,100]$. (a): The microscopic state $u(x,t)$ in
  the active region (left) is similar to the one found for the single bump (see
  Figure~\ref{fig:exampleBump}(a)). A comparison between $J(x,50)$ and $u(x,50)$ is
  reported on the right panel, where we also mark the intervals $[\xi_1, \xi_2],
  \ldots, [\xi_7,\xi_8]$ where $J$ is above the firing threshold $h$. (b): Space-time
  plots of $u$ and $J$. Parameters are as in Table~\ref{tab:params}.}
  \label{fig:exampleMultiBump}
\end{figure}

\subsection{Multiple-bumps solutions} \label{subsec:multiBumps}
Solutions with multiple bumps are also observed by direct simulation,
as shown in Figure~\ref{fig:exampleMultiBump}. The microstructure of these
patterns
resembles the one found in the single bump case (see
Figure~\ref{fig:exampleMultiBump}(a)). At the mesoscopic level, the set for which $J$
lies above the threshold $h$ is now a union of disjoint intervals $[\xi_1,\xi_2],
\dots, [\xi_7,\xi_8]$. The number of bumps of the pattern depends on the width of the
tissue; the experiment of Figure~\ref{fig:exampleMultiBump} is carried out on a
domain twice as large as that of Figure~\ref{fig:exampleBump}. The examples of
bump and multiple-bump solutions reported in these figures are obtained for different
values of the main control parameter $\kappa$ (see Table~\ref{tab:params}), however,
these states coexist in a suitable region of parameter space, as will be shown
below.

\begin{figure}
  \centering
  \includegraphics{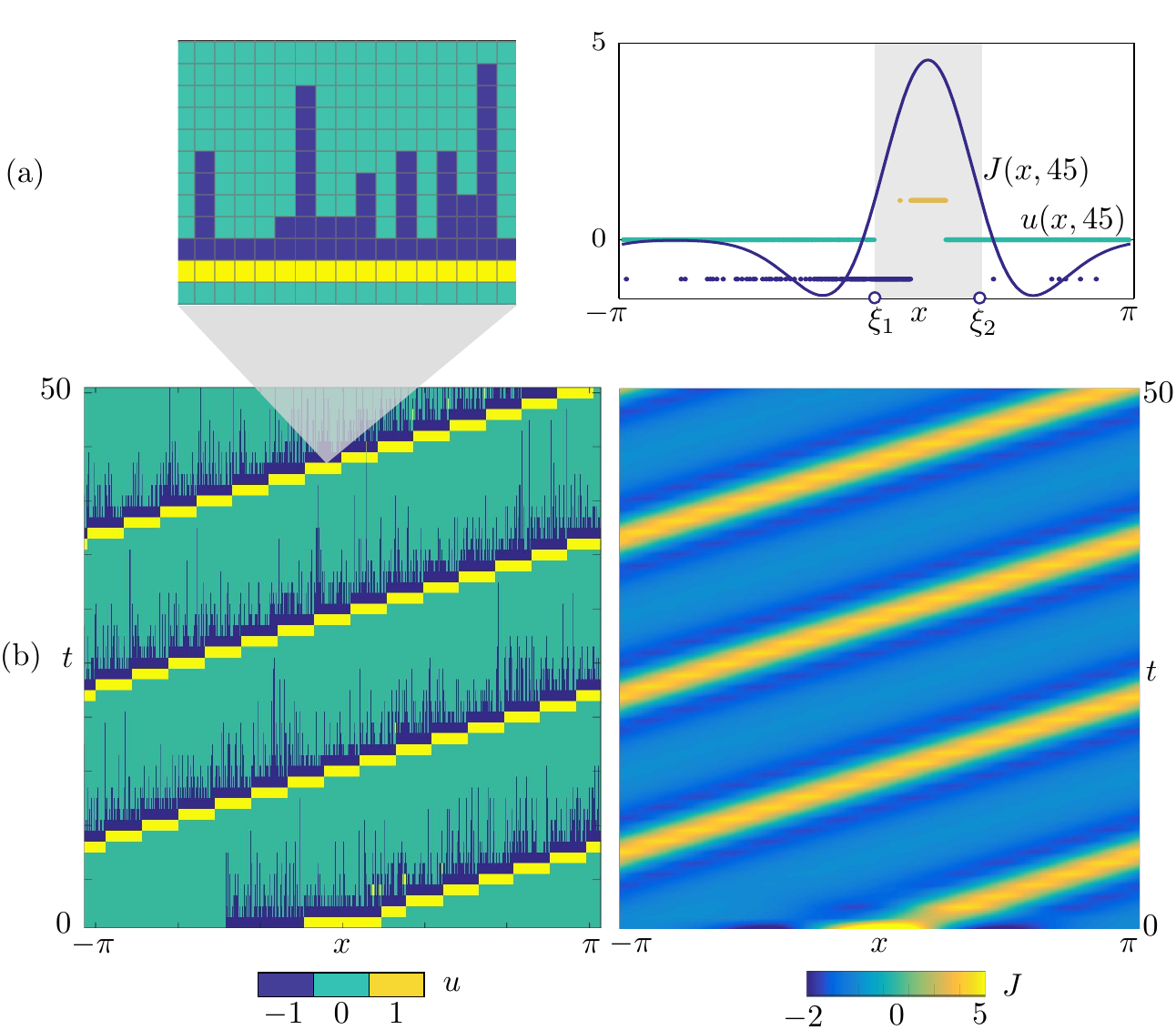}
  \caption{Travelling wave obtained via time simulation of the stochastic model for
  $(x,t) \in [-\pi,\pi] \times [0,50]$. (a): The microscopic state $u(x,t)$
  (left) has a characteristic microstructure, which is also visible on the right
  panel, where we compare $J(x,45)$ and $u(x,45)$. As in the other cases, we mark the
  interval $[\xi_1, \xi_2]$ where $J$ is above the firing threshold $h$. (b):
  Space-time plots of $u$ and $J$. Parameters are as in Table~\ref{tab:params}.}
  \label{fig:exampleWave}
\end{figure}

\subsection{Travelling waves} \label{subsec:travellingWaves}
Further simulation shows that the model also supports coherent states in the
form of stochastic travelling waves. In two spatial dimensions, the system is known to
support travelling spots~\cite{Gong2012aa,Qi2015aa}. In Figure~\ref{fig:exampleWave},
we show a time simulation of the stochastic model with initial condition
\[
u(x,0) = \sum_{k \in \USet} k \ind_{X_k}(x) 
\quad \text{with partition} \quad
  \begin{aligned}
    X_{-1} & = [-1.5,-0.5), \\
    X_{0}  & = [-\pi,-1.5)\cup[0.5,\pi), \\
    X_{1}  & = [-0.5,0.5).
\end{aligned}
\]
In passing, we note that the state of the network at each discrete time $t$ is defined
entirely by the partition $\{X_k\}$ of the tissue; we shall often use this
characterisation in the reminder of the paper.
\begin{figure}
  \centering
  \includegraphics{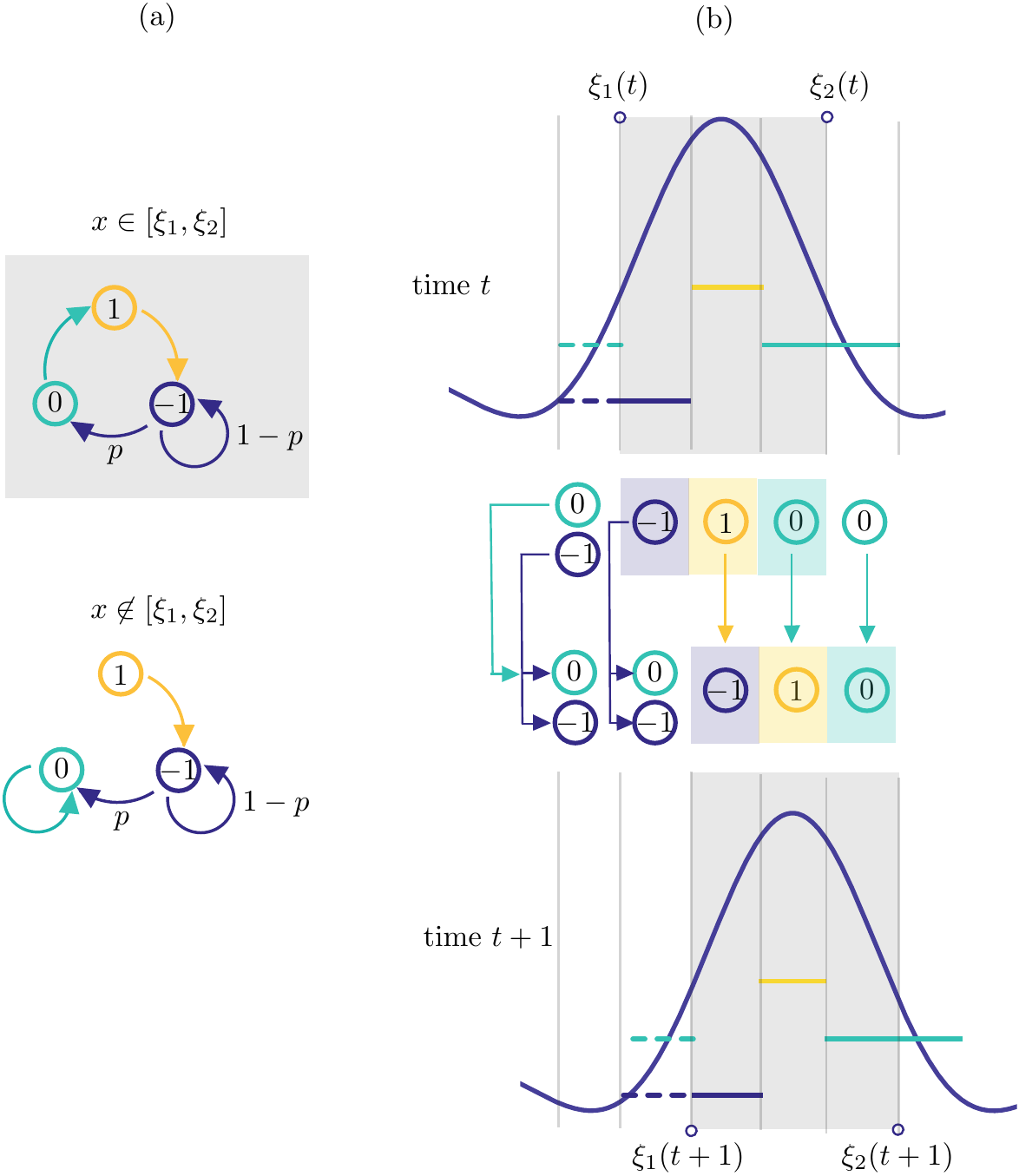}
  \caption{Schematic of the advection mechanism for the travelling wave state. Shaded
  areas pertain to the active region $[\xi_1(t),\xi_2(t)]$, non-shaded areas
  to the inactive region $\XSet \setminus [\xi_1(t),\xi_2(t)]$. (a): In the
  active (inactive) region, $q_i = f(J(u))(x_i) \approx 1$ ($q_i \approx 0$),
  hence the transition kernel~\eqref{eq:probRefr}--\eqref{eq:probSpike} can be simplified
  as shown. (b): At time $t$ the travelling wave has a profile similar to the one in
  Figure~\ref{fig:exampleWave}, which we represent in 
  the proximity of the active region. We depict 5 intervals of equal width, $3$
  of which form a partition of $[\xi_1(t),\xi_2(t)]$. Each interval is mapped to
  another interval at time $t+1$, following the transition rules sketched in
  (a). In one discrete step, the wave progresses with positive speed: so that
  $J(x,t+1)$ is a translation of $J(x,t)$.}
  \label{fig:advectionMechanism}
\end{figure}

In the direct simulation of Figure~\ref{fig:exampleWave}, the active region moves to
the right and, after just $4$ iterations, a travelling
wave emerges. The microscopic variable, $u(x,t)$, displays stochastic fluctuations
which disappear at the level of the mesoscopic variable, $J(x,t)$, giving rise to a
seemingly deterministic travelling wave. A closer inspection
(Figure~\ref{fig:exampleWave}(a)) reveals that the state can still be described
in terms of the active region $[\xi_1, \xi_2]$ where $J$ is above $h$. However, the
travelling wave has a different microstructure with respect to the bump. Proceeding
from right to left, we observe: 
\begin{enumerate}
  \item A region of the tissue ahead of the wave, $x \in (\xi_2,\pi)$, where the
    neurons are in the quiescent state $0$ with high probability.
  \item An active region $x \in [\xi_1,\xi_2]$, split in three subintervals, each of
    approximate width $(\xi_2 - \xi_1)/3$, where $u$ attains with high probability
    the values $0$, $1$ and $-1$ respectively. 
  \item A region at the back of the wave, $ x \in [-\pi,\xi_1)$, where neurons are
    either quiescent or refractory. We note that $u=0$ with high probability as $x \to
    -\pi$ whereas, as $x \to \xi_1$, neurons are increasingly likely to be
    refractory, with $u=-1$.
\end{enumerate}

A further observation of the space-time plot of $u$ in
Figure~\ref{fig:exampleWave}(b) reveals a remarkably simple advection mechanism of
the travelling wave, which can be fully understood in terms of the transition
kernel of Figure~\ref{fig:schematic}(b) upon noticing that, for sufficiently large
$\beta$, $q_i = f(J(u))(x_i) \approx
0$ everywhere except in the active region, where $q_i \approx 1$. In
Figure~\ref{fig:advectionMechanism}, we show how the transition kernel simplifies
inside and outside the active region and provide a schematic of the advection
mechanism. For an idealised travelling wave profile at time $t$, we depict 3
subintervals partitioning the active region (shaded), together with 2 adjacent
intervals outside the active region. Each interval is then mapped to another
interval, following the simplified transition rules sketched in
Figure~\ref{fig:advectionMechanism}(a): 
\begin{enumerate}
  \item At the front of the wave, to the right of $\xi_2(t)$, neurons in the
    quiescent state $0$ remain at $0$ (rules for $x \not \in [\xi_1,\xi_2]$).
  \item Inside the active region, to the left of $\xi_2(t)$, we follow
    the rules for $x \in [\xi_1,\xi_2]$ in a clockwise manner: neurons in
    the quiescent state $0$ spike, hence their state variable becomes $1$;
    similarly, spiking neurons become refractory.
    Of the neurons in the refractory state, those being the
    ones nearest $\xi_1(t)$, a proportion $p$ become
    quiescent, while the remaining ones remain refractory.
  \item At the back of the wave, to the left of $\xi_1(t)$, the interval
    contains a mixture of neurons in states $0$ and $-1$. The former
    remain at $0$ whilst, of the latter, a proportion $p$ transition
    into state $0$, with the rest remaining at $-1$ (rules for $x
    \not \in [\xi_1,\xi_2]$). From this argument, we see that the proportion of
    refractory neurons in the back of the wave must decrease as $\xi \to -\pi$.
\end{enumerate}
The resulting mesoscopic variable $J(x,t+1)$ is a spatial translation by $( \xi_2(t) -
\xi_1(t) )/3$ of $J(x,t)$. We remark that the approximate transition rules
of Figure~\ref{fig:advectionMechanism}(a) are valid also in the case of a bump,
albeit the corresponding microstructure does not allow the advection mechanism
described above.

\begin{figure}
  \centering
  \includegraphics{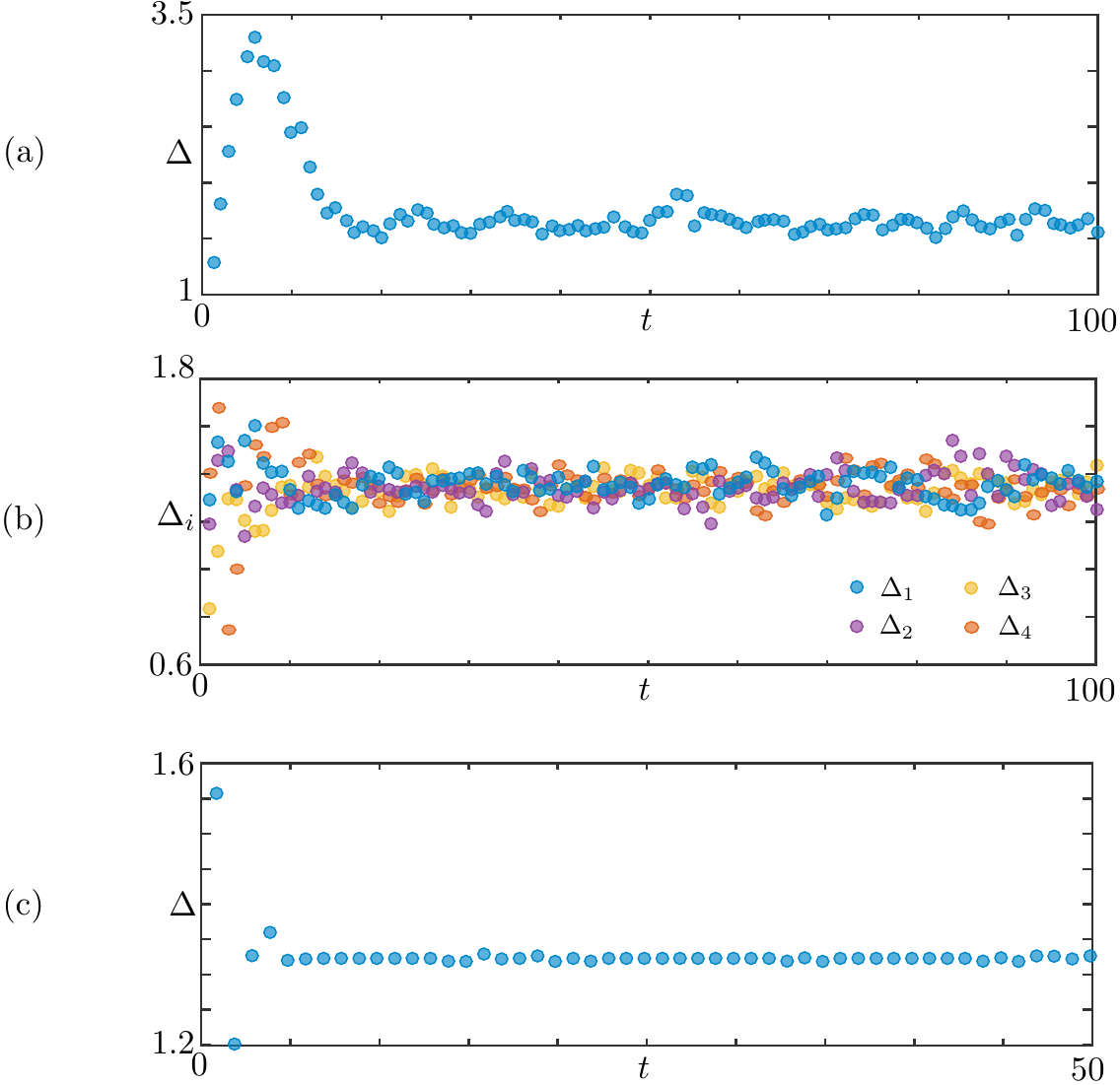}
  \caption{Width of the active regions $\Delta_i = \xi_{2i} - \xi_{2i-1}$ for the
  patterns in Figures~\ref{fig:exampleBump}--\ref{fig:exampleWave}. (a): Bump, for
  which $i=1$. (b): Multiple Bump, $i=1,\ldots,4$. (c): Travelling wave, $i=1$. In
  all cases, the patterns reach a coarse equilibrium state after a short transient.}
  \label{fig:slaving}
\end{figure}

\subsection{Macroscopic variables}\label{sec:macroscopicDescriptions}
The computations of the previous sections suggest that, beyond the mesoscopic
variable, $J(x)$, coarser macroscopic variables are available to describe the observed
patterns. In analogy with what is typically found in neural fields with Heaviside
firing rate~\cite{Amari1977aa,Bressloff:2014cm,Coombes2004aa},
the scalars $\{\xi_i\}$
defining the active region $X_\geq = \cup_i [\xi_{2i-1}, \xi_{2i}]$, where $J$ is
above $h$, seem plausible macroscopic variables. This is evidenced not only by
Figures~\ref{fig:exampleBump}--\ref{fig:exampleWave}, but also from the schematic in
Figure~\ref{fig:advectionMechanism}(b), where the interval $[\xi_1(t),\xi_2(t)]$ is
mapped to a new interval $[\xi_1(t+1),\xi_2(t+1)]$ of the same width. To explore this
further, we extract the widths $\Delta_i(t)$ of each sub-interval
$[\xi_{2i}(t),\xi_{2i-1}(t)]$ 
from the data in Figure~\ref{fig:exampleBump}--\ref{fig:exampleWave}, and plot the
widths as a function of $t$. In all cases, we observe a brief transient, after which
$\Delta_i(t)$ relaxes towards a coarse equilibrium, though fluctuations seem larger
for the bump and multiple bump when compared with those for the wave. In the multiple
bump case, we also notice that all intervals have approximately the same asymptotic
width. 

\section{Deterministic model} \label{sec:macroModel}
We now introduce a deterministic version of the stochastic model considered in
Section~\ref{sec:microModel}, which is suitable for carrying out analytical
calculations. We make the following assumptions:
\begin{enumerate}
  \item \textit{Continuum neural tissue.} We consider the limit of infinitely
    many neurons and pose the model on $\SSet$.
  \item \textit{Deterministic transitions.} We assume $p=1$, which
    implies a deterministic transition from refractory states to quiescent ones (see
    Equation~\eqref{eq:probRefr}), and $\beta \to \infty$, which induces a
    Heaviside firing rate $f(I) = \Theta(I-h)$ and hence a deterministic transition
    from quiescent states to spiking ones given sufficiently high input (see Equations~\eqref{eq:firingRate},
    \eqref{eq:probQuies}).
\end{enumerate}

In addition to the pullback sets $X_{-1}$, $X_0$, and $X_1$ defined
in~\eqref{eq:partitions}, we will partition the tissue into
\textit{active} and \textit{inactive} regions
\begin{equation} \label{eq:actSets}
  X_\geq(t) = \Set{ x \in \XSet \colon J(x,t) \geq h }, 
  \qquad 
  X_<(t) = \XSet \setminus X_\geq(t).
\end{equation}

In the deterministic model, the
transitions~\eqref{eq:probRefr}--\eqref{eq:probSpike} are then replaced by the
following rule
\begin{equation}\label{eq:detModel}
u(x,t+1) = 
\begin{cases}
  -1 & \textrm{if $x \in X_1(t)$}, \\
  0 & \textrm{if $x \in X_{-1}(t) \cup \big( X_0(t) \cap X_<(t) \big) $}, \\
  1 & \textrm{if $x \in X_0(t) \cap X_\geq (t) $}.
\end{cases}
\end{equation}
We stress that the right-hand side of the equation above depends on $u(x,t)$, since
the partitions $\{ X_{-1}, X_0, X_1 \}$ and $\{ X_< , X_\geq \}$ do so 
(see Remark~\ref{rem:uDependence}). 

As we shall see, it is sometimes useful to refer to the induced mapping of
the pullback sets
\begin{equation} \label{eq:pullbackModel}
  \begin{aligned}
    X_{-1}(t+1) & = X_1(t) \\
    X_0(t+1)    & = X_{-1}(t) \cup \big( X_0(t) \cap X_<(t) \big) \\
    X_1(t+1)    & = X_0(t) \cap X_\geq (t) 
  \end{aligned}
  \;.
\end{equation}
Henceforth, we will use the term \textit{deterministic model} and formally write
\begin{equation}\label{eq:detPhi}
  u(x,t+1) = \Phidet(u(x,t); \gamma).
\end{equation}
for 
\eqref{eq:detModel},
where the partition $\{ X_k \}_{k \in \USet}$ is defined
by~\eqref{eq:partitions} and the active and inactive sets $X_\geq$, $X_<$
by~\eqref{eq:actSets}.

\section{Macroscopic bump solution of the deterministic model}\label{sec:MacroBump}
\begin{figure}
  \centering
  \includegraphics{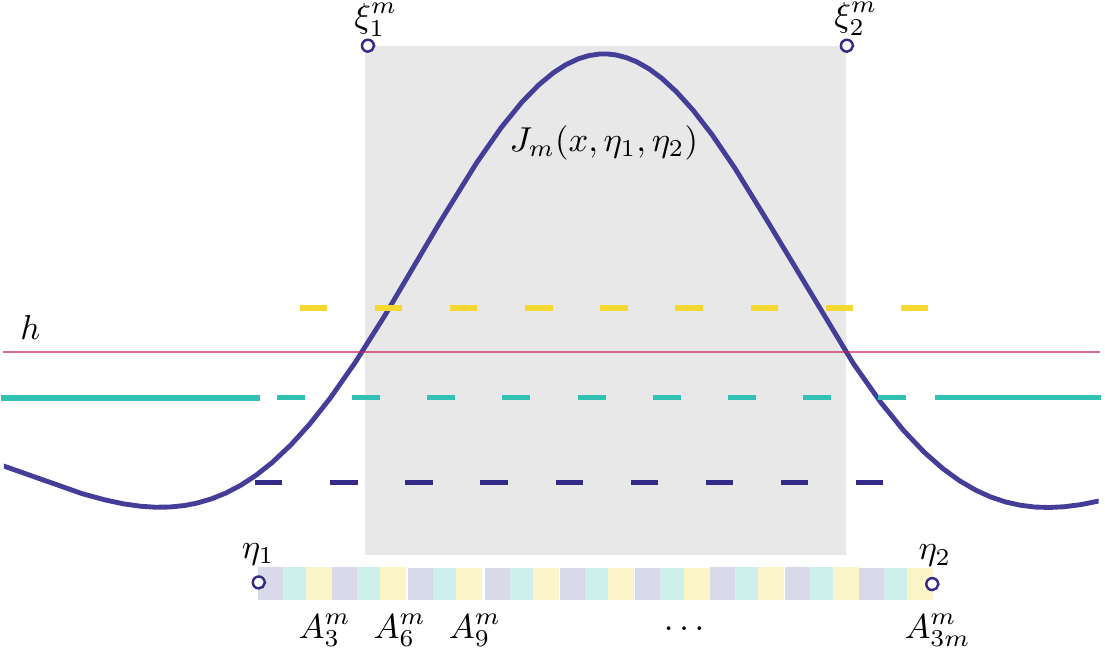}
  \caption{Schematic of the analytical construction of a bump. A microscopic state
  whose partition comprises $3m+2$ strips is considered. The microscopic
  state, which is not an equilibrium of the deterministic system, has a
  characteristic width $\eta_2 - \eta_1$, which differs from the width $\xi^m_2
  -\xi^m_1$ of the mesoscopic bump $J_m$. If we let $m \to \infty$ while keeping
  $\eta_2 - \eta_1$ constant, then $J_m$ tends towards a mesoscopic bump
  $J_\mathrm{b}$ and $\xi^m_i \to \eta_i$ (see
  Proposition~\ref{prop:macroBump}).} \label{fig:analyticBump}
\end{figure}

We now proceed to construct a bump solution of the deterministic model presented in
Section~\ref{sec:macroModel}. In order to do so, we consider a microscopic state
with a regular structure, resulting in a partition, $\{X_k^m\}_k$, with $3m+2$
strips (see Figure~\ref{fig:analyticBump}) and then study the limit $m\to \infty$.

\subsection{Bump construction}\label{subsec:macroBump}
Starting from two points $\eta_1, \eta_2 \in \SSet$, with $\eta_1< \eta_2$, we
construct $3m$ intervals as follows
\begin{equation}\label{eq:Am}
  A^m_i = \bigg[ \eta_1 + \frac{i-1}{3m} (\eta_2-\eta_1) , \eta_1 
  + \frac{i}{3m} (\eta_2 - \eta_1) \bigg), 
  \qquad i =1,\ldots,3m, \quad m \in \NSet.
\end{equation}
We then consider states $u_m(x) =  \sum_{k \in \USet} k \ind_{X^m_{k}}(x)$, with
partitions given by
\begin{equation}\label{eq:Xm}
  X^m_{-1} = \bigcup_{j = 1}^{m} A^m_{3j-2}, \quad
  X^m_0 = [-L,\eta_1) \cup [\eta_2,L) \bigcup_{j = 1}^{m} A^m_{3j-1}, \quad
  X^m_1 = \bigcup_{j = 1}^{m} A^m_{3j}, 
\end{equation}
and activity set $X_\geq = [\xi^m_1, \xi^m_2]$. We note that, in addition to the
$3m$ strips that form the active region of the bump, we also need two additional
strips in the inactive region to form a partition of $\SSet$.
In general, $\{\xi^m_i\}_i \neq \{\eta_i\}_i$, as
illustrated in Figure~\ref{fig:analyticBump}. Applying~\eqref{eq:detModel}
or~\eqref{eq:pullbackModel}, we find $\Phi_\textrm{d}(u_m) \neq u_m$, hence
$u_m$ are not equilibria of the deterministic model. However, these states help us
defining a macroscopic bump as a fixed point of a suitably defined map using the
associated mesoscopic synaptic profile
\begin{equation}\label{eq:Jm}
  J_m(x,\eta_1,\eta_2) = \kappa \int_{X_1^m(\eta_1,\eta_2)} W(x-y)\, \d y,
\end{equation}
where we have highlighted the dependence of $X_1^m$ on $\eta_1,\eta_2$. The
proposition below shows that there is a well defined limit, $J_\textrm{b}$, of the
mesoscopic profile as $m \to \infty$. We also have that $\xi^m_i \to \eta_i$ as
$m\to\infty$ and that the threshold crossings of the activity set are roots
of a simple nonlinear function.
\begin{prop}[Bump construction] \label{prop:macroBump}
  Let $W$ be the periodic extension of the synaptic kernel~\eqref{eq:kernel} and
  let $h, \kappa \in \RSet_+$. Further, let $\{ A_i^m \}_{i=1}^{3m}$, $X^m_1$
  and $J_m$ be defined as in~\eqref{eq:Am}, \eqref{eq:Xm} and \eqref{eq:Jm},
  respectively, and let $J_\mathrm{b} \colon \SSet^3 \to \RSet$ be defined as
  \[
  J_\mathrm{b}(x,\eta_1,\eta_2) = \frac{\kappa}{3} \int_{\eta_1}^{\eta_2} W(x-y)\, dy.
  \]
  The following results hold
  \begin{enumerate}
    \item $J_m(x,\eta_1,\eta_2) \to J_\mathrm{b}(x, \eta_1, \eta_2)$ as $m \to
      \infty$ uniformly
      in the variable $x$ for all $\eta_1, \eta_2 \in \SSet$ with
      $\eta_1 < \eta_2$,
    \item If there exists $\Delta \in (0,L)$ such that
      $3h = \kappa \int_0^\Delta W(y)\, dy$, then
      \[
      J_\mathrm{b}(0,0,\Delta) = J_\mathrm{b}(\Delta,0,\Delta) = h.
      \]
  \end{enumerate}
\end{prop}
\begin{proof}
  We fix $\eta_1<\eta_2$ and consider the $2L$-periodic continuous mapping $ x \mapsto
  J_\mathrm{b}(x,\eta_1,\eta_2)$, defined on $\SSet$. We aim to prove that $J_m \to
  J_\mathrm{b}$
  uniformly in $\SSet$. We pose
  \begin{align*}
    & I^m_{-1}(x) = \sum_{j=1}^m \int_{A_{3j-2}} W(x-y)\, \d y, \\
    & I^m_0 (x)  = \sum_{j=1}^m \int_{A_{3j-1}} W(x-y)\, \d y, \\
    & I^m_1 (x)  = \sum_{j=1}^m \int_{A_{3j}} W(x-y)\, \d y,
  \end{align*}
  for all $x \in \SSet$, $m \in \NSet$. Since the intervals $\{ A_i^m \}_{i=1}^{3m}$
  form a partition of $[\eta_1,\eta_2)$ we have
  \begin{equation}\label{eq:sumInt}
    \frac{3}{\kappa} J_\mathrm{b}(x) = I^m_{-1}(x) + I^m_{0}(x) + I^m_{1}(x) \quad 
    \text{for all $x \in \SSet$, $m \in \NSet$}.
  \end{equation}
  Since $W$ is continuous on the compact set $\SSet$, it is also uniformly continuous
  on $\SSet$. Hence, there exists a modulus of continuity $\omega$ of $W$:
  \[
  \omega(r) = \sup_{
		  \substack{ p,q \in \SSet \\
		  |p - q| \leq r} }
		  | W(p) - W(q) |, 
   \qquad
   \textrm{with}
   \lim_{r \to 0^+} \omega(r) = \omega(0) = 0.
  \]
  We use $\omega$ to estimate $\vert I_1^m (x) - I_0^m (x) \vert$ as follows:
  \[
  \begin{split}
    \vert I_1^m (x) - I_0^m (x) \vert 
    & \leq
      \sum_{j=1}^m
      \bigg\vert 
	   \int_{A_{3j}} W(x-y) \, dy
	  - \int_{A_{3j-1}} W(x-y) \, dy
      \bigg\vert \\
    & = \sum_{j=1}^m
      \bigg\vert 
      \int_{\eta_1 + \frac{3j-1}{3m} (\eta_2-\eta_1)}
          ^{\eta_1 + \frac{3j}{3m} (\eta_2-\eta_1)} W(x-y) \, dy
      -
      \int_{\eta_1 + \frac{3j-2}{3m} (\eta_2-\eta_1)}
          ^{\eta_1 + \frac{3j-1}{3m} (\eta_2-\eta_1)} W(x-y) \, dy
      \bigg\vert \\
    & = \sum_{j=1}^m
      \bigg\vert 
      \int_{\eta_1 + \frac{3j-1}{3m} (\eta_2-\eta_1)}
          ^{\eta_1 + \frac{3j}{3m} (\eta_2-\eta_1)}
          W(x-y) - 
	  W\bigg(x-y+ \frac{\eta_2 - \eta_1}{3m} \bigg)
	  \, dy
      \bigg\vert \\
    & \leq \sum_{j=1}^m
      \int_{\eta_1 + \frac{3j-1}{3m} (\eta_2-\eta_1)}
          ^{\eta_1 + \frac{3j}{3m} (\eta_2-\eta_1)}
	  \bigg\vert W(x-y) - 
	  W\bigg(x-y+ \frac{\eta_2 - \eta_1}{3m} \bigg) 
	  \bigg\vert
      \, dy \\
    & \leq \sum_{j=1}^m
      \int_{\eta_1 + \frac{3j-1}{3m} (\eta_2-\eta_1)}
          ^{\eta_1 + \frac{3j}{3m} (\eta_2-\eta_1)}
      \omega\bigg(\frac{\eta_2 - \eta_1}{3m} \bigg) \, dy \\
    & =  \omega\bigg(\frac{\eta_2 - \eta_1}{3m} \bigg) \frac{\eta_2-\eta_1}{3} .
  \end{split}
  \]
  We have then $\vert I_1^m (x) - I_0^m (x) \vert \to 0$ as $m \to \infty$ and
  since $\omega\big( (\eta_2-\eta_1)/(3m) \big)$ is independent of $x$, the
  convergence is uniform. Applying a similar argument, we find $\vert I_{-1}^m
  (x) - I_0^m (x) \vert \to 0$ as $m \to \infty$ and using~\eqref{eq:sumInt}, we
  conclude $I^m_1, I^m_0, I^m_{-1} \to J_\mathrm{b}/\kappa$ as $m \to \infty$.
  Since $I^m_1 = J_m/\kappa$, then $J_m \to J_\mathrm{b}$ uniformly for all $x \in \SSet$ and
  $\eta_1,\eta_2 \in \SSet$ with $\eta_1 < \eta_2$, that is, result 1 holds true. 

  By hypothesis $J_\mathrm{b}(0,0,\Delta) = h$ and, using a change of variables under the
  integral and the fact that $W$ is even, it can be shown that $J_\mathrm{b}(\Delta,0,\Delta)
  = h$, which proves result 2.
\end{proof}

\begin{cor}[Bump symmetries]
  Let $\Delta$ be defined as in Proposition~\ref{prop:macroBump}, then
  $J_\mathrm{b}(x+\delta,\delta, \delta+\Delta)$ is a mesoscopic bump for all $\delta \in
  [L, -\Delta + L)$. Such bump is symmetric with respect to the
  axis $x=\delta + \Delta/2$.
\end{cor}
\begin{proof}
  The assertion is obtained using a change of variables in the integral
  defining $J_\mathrm{b}$ and noting that $W$ is even.
\end{proof}

The results above show that, $\xi^m_i \to \eta_i$ as $m \to \infty$, hence we lose
the distinction between width of the microscopic pattern, $\eta_2-\eta_1$, and width
of the mesoscopic pattern, $\xi^m_2 - \xi^m_1$, in that result 2 establishes
$J_\mathrm{b}(\eta_i,\eta_1,\eta_2)=h$, for $\eta_1=0$, $\eta_2=\Delta$. With reference to
Figure~\ref{fig:analyticBump}, the factor $1/3$ appearing in the expression 
for $J_\mathrm{b}$
confirms that, in the limit of infinitely many strips, only a third of the
intervals $\{A_j^m\}_j$ contribute to the integral. In addition, the formula
for $J_\mathrm{b}$ is useful for practical computations as it allows us
to determine the width, $\Delta$, of
the bump.

\begin{rem}[Permuting intervals $A^m_i$]\label{rem:perm}
  A bump can also be found if the partition $\{X^m_k\}$ of $u_m$ is less
  regular than the one depicted in Figure~\ref{fig:analyticBump}. In
  particular, Proposition~\ref{prop:macroBump} can be extended to a more general case of
  permuted intervals. More precisely, if we consider permutations, $\sigma_j$, of the
  index sets $\Set{3j-2,3j-1,3j}$ for $j=1,\ldots,m$ and construct partitions
  \[
  X^m_{-1} = \bigcup_{j = 1}^{m} A^m_{\sigma_j(3j-2)}, \quad
  X^m_0 = [-L,0) \cup [\Delta,L) \bigcup_{j = 1}^{m} A^m_{\sigma_j(3j-1)}, \quad
  X^m_1 = \bigcup_{j = 1}^{m} A^m_{\sigma_j(3j)},
  \]
  then the resulting $J_m$ converges uniformly to $J_\mathrm{b}$ as $m \to \infty$. The proof
  of this result follows closely the one of Proposition~\ref{prop:macroBump} and is
  omitted here for simplicity.
\end{rem}

\subsection{Bump stability}
Once a bump has been constructed, its stability can be studied by employing standard
techniques used to analyse neural field models~\cite{Bressloff2012o}. We consider the map
\[
\Psi_\mathrm{b} \colon \SSet^2 \times \SSet^2 \to \RSet^2, \qquad
( \xi, \eta )
\mapsto
\begin{bmatrix}
  J_\mathrm{b}(\xi_1, \eta_1, \eta_2) - h \\
  J_\mathrm{b}(\xi_2, \eta_1, \eta_2) - h 
\end{bmatrix}
.
\]
and study the implicit evolution 
\begin{equation}\label{eq:phiBump}
  \Psi_\mathrm{b}(\xi(t+1),\xi(t)) = 0.
\end{equation}
The motivation for studying this evolution comes from
Proposition~\ref{prop:macroBump}, according to which the macroscopic bump
$\xi_*=(0,\Delta)$ is an equilibrium of~\eqref{eq:phiBump}, that is, $\Psi_\mathrm{b}(\xi_*,\xi_*) = 0$.
To determine coarse linear stability, we study how small perturbations of $\xi_*$
evolve according to the implicit rule~\eqref{eq:phiBump}.
We set $\xi(t) = \xi_* + \eps \widetilde{\xi}(t)$, for $0 < \epsilon \ll 1$ with $\widetilde{\xi}_i 
= \mathcal{O}(1)$ and expand~\eqref{eq:phiBump} around $(\xi_*,\xi_*)$, retaining
terms up to order
$\eps$,
\[
\Psi_\mathrm{b}(\xi_*+\eps \widetilde{\xi}(t+1),
     \xi_*+\eps \widetilde{\xi}(t)) =
     \Psi_\mathrm{b}(\xi_*, \xi_*) + \eps D_\xi \Psi_\mathrm{b}(\xi_*, \xi_*) \widetilde{\xi}(t+1) + 
                           \eps D_\xi \Psi_\mathrm{b}(\xi_*, \xi_*) \widetilde{\xi}(t). 
\]
Using the classical ansatz $\tilde \xi(t) = \lambda^t v$, with $\lambda \in
\CSet$ and $v\in\SSet^2$, we obtain the eigenvalue problem
\begin{equation}\label{eq:macroBumpStab}
  \lambda
\begin{bmatrix}
          v_1 \\
          v_2 
\end{bmatrix}
 =
 \frac{1}{W(0)-W(\Delta)}
\begin{bmatrix}
          -W(0)      & W(\Delta) \\
           W(\Delta) & -W(0) \\
\end{bmatrix}
\begin{bmatrix}
          v_1 \\
          v_2 
\end{bmatrix},
\end{equation}
with eigenvalues and eigenvectors given by
\begin{align*}
  & \lambda_1 = \frac{W(\Delta) - W(0)}{ W(0) - W(\Delta)}=-1, 
               & v^1 = (1,1)^T, \\
  & \lambda_2 = \frac{-W(0) - W(\Delta)}{ W(0) - W(\Delta)}, 
	       & \qquad v^2 = (-1,1)^T.
\end{align*}
As expected, we find an eigenvalue with absolute value equal to $1$,
corresponding to a pure translational eigenvector. The remaining eigenvalue,
corresponding to a compression/extension eigenvector, determines the stability
of the macroscopic bump. The parameters $A_i$, $B_i$ in Equation~\eqref{eq:kernel} 
are such that $W$ has a global maximum at $x=0$, with $W(0)>0$. Hence, the eigenvalues
are finite real numbers and the pattern is stable if $W(\Delta)<0$. We will present
concrete bump computations in Section~\ref{sec:numerics}.

\subsection{Multi-bump solutions}
The discussion in the previous section can be extended to the case of solutions
featuring multiple bumps. For simplicity, we will discuss here solutions
with $2$ bumps, but the case of $k$ bumps follows straightforwardly. The starting
point is a microscopic structure similar to~\eqref{eq:Xm}, with two disjoint
intervals $[\eta_1,\eta_2),[\eta_3, \eta_4) \subset \SSet$ each subdivided into $3m$
subintervals. We form the vector $\eta=\{\eta_i\}_{i=1}^4$ and have
\[
  \kappa \int_{X^m_1} W(x-y)\, dy 
  = 
  \sum_{j=1}^2 J_{m}(x,\eta_{2j-1},\eta_{2j})  
  \to \sum_{j=1}^2 J_\mathrm{b}(x,\eta_{2j-1},\eta_{2j}),
\]
as $m \to \infty$ uniformly in the variable $x$ for all $\eta_1, \ldots, \eta_4 \in
\SSet$ with $\eta_1 < \ldots < \eta_4$. In the expression above, $J_m$ and
$J_\mathrm{b}$ are
the same functions used in Section~\ref{subsec:macroBump} for the single bump. In analogy
with what was done for the single bump, we consider the mapping defined by
\[
\Psi \colon \SSet^4 \times \SSet^4 \to \RSet^4, \qquad
( \xi, \eta )
\mapsto
\bigg\{
  -h + \sum_{j=1}^2 J_\mathrm{b}(\xi_i, \eta_{2j-1}, \eta_{2j})
\bigg\}_{i = 1}^4
.
\]

Multi-bump solutions can then be studied as in Section~\ref{sec:MacroBump}. 
We present here the results for a multi-bump for $L=\pi$ with threshold crossings
given  by
\begin{equation}\label{eq:multiXi}
\xi_* = 
\frac{1}{2}
\begin{bmatrix}
-\pi - \Delta \\
-\pi + \Delta \\
 \pi - \Delta \\
 \pi + \Delta \\
\end{bmatrix},
\end{equation}
where $\Delta$ satisfies
\begin{equation}\label{eq:multiDeltaEq}
J_\mathrm{b}\bigg(\frac{\pi+\Delta}{2},
\frac{
-\pi - \Delta
}{2}
,
\frac{
-\pi + \Delta
}{2}
\bigg) 
+
J_\mathrm{b}\bigg(\frac{\pi+\Delta}{2},
\frac{
\pi - \Delta
}{2}
,
\frac{
\pi + \Delta
}{2}
\bigg) 
= h.
\end{equation}
\begin{figure}
  \centering
  \includegraphics{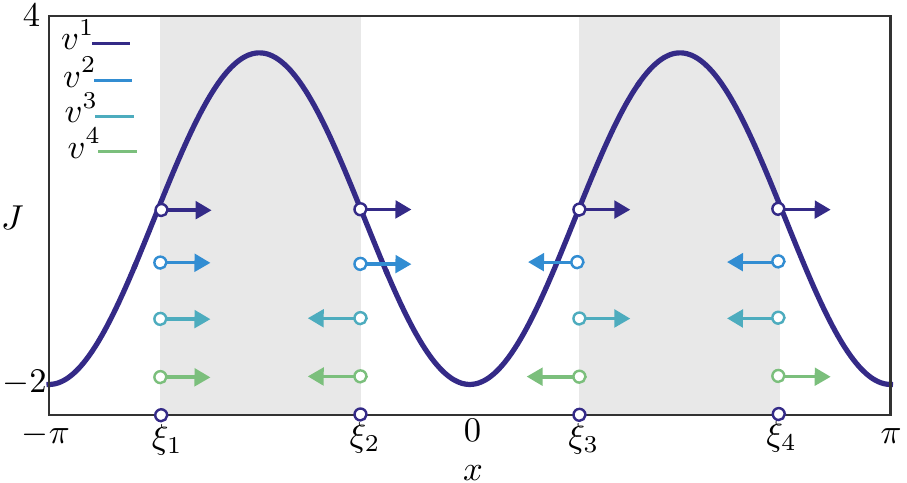}
  \caption{Stable mesoscopic multi-bump obtained for the deterministic model. We also
  plot the corresponding macroscopic bump $\xi_*$
  (Equations~\eqref{eq:multiXi}--\eqref{eq:multiDeltaEq}) and coarse eigenvectors.
  Parameters are $\kappa =
  30$, $h=0.9$, $p=1$, $\beta \to \infty$, with other parameters as in
  Table~\ref{tab:params} }
  \label{fig:multiBumpEigenvectors}
\end{figure}
A quick calculation leads to the eigenvalue problem
\begin{equation}\label{eq:multiBumpStab}
  \lambda
\begin{bmatrix}
          v_1 \\
          v_2 \\
          v_3 \\
          v_4
\end{bmatrix}
 =
 \frac{1}{\alpha}
\begin{bmatrix}
W(0)       & -W(\Delta) & W(\pi) & -W(\pi - \Delta) \\
-W(\Delta) & W(0)       &-W(\pi-\Delta) & W(\pi) \\
W(\pi) & -W(\pi-\Delta) & W(0)  &-W(\Delta) \\
-W(\pi-\Delta) &  W(\pi) & -W(\Delta)  & W(0) \\
\end{bmatrix}
\begin{bmatrix}
          v_1 \\
          v_2 \\
          v_3 \\
          v_4
\end{bmatrix},
\end{equation}
where $\alpha = -W(0) + W(\Delta) - W(\pi) + W(\pi-\Delta)$. The real symmetric
matrix in Equation~\eqref{eq:multiBumpStab} has eigenvalues and eigenvectors given by
\begin{align*}
  & \lambda_1 = \frac{-W(0) + W(\Delta)-W(\pi) + W(\pi-\Delta)}
  { W(0) - W(\Delta) + W(\pi) - W(\pi-\Delta)}=-1, 
               & & v^1 = (1,1,1,1)^T, \\
  & \lambda_2 = \frac{-W(0) + W(\Delta) + W(\pi) - W(\pi-\Delta)}{ W(0) - W(\Delta) + W(\pi) - W(\pi-\Delta)}, 
               & & v^2 = (1,1,-1,-1)^T, \\
  & \lambda_3 = \frac{-W(0)-W(\Delta)-W(\pi)-W(\pi-\Delta)}{ W(0) - W(\Delta) + W(\pi) - W(\pi-\Delta)},
               & & v^3 = (1,-1,1,-1)^T, \\
  & \lambda_4 = \frac{-W(0) - W(\Delta)+W(\pi) + W(\pi-\Delta)}
  { W(0) - W(\Delta) + W(\pi) - W(\pi-\Delta)}, 
               & & v^4 = (1,-1,-1,1)^T. 
\end{align*}
As expected, we have one neutral translational mode. If the remaining 3 eigenvalues
lie in the unit circle, the multi-bump solution is stable. A depiction of this
multi-bump, with corresponding eigenmodes can be found in
Figure~\ref{fig:multiBumpEigenvectors}. We remark that the multi-bump
presented here was constructed imposing particular symmetries (the pattern is
even; bumps all have the same widths). The system may in principle support more generic
bumps, but their construction and stability analysis can be carried out in a similar
fashion.

\section{Travelling waves in the deterministic model} \label{sec:TWDet}
\begin{figure}
  \centering
  \includegraphics{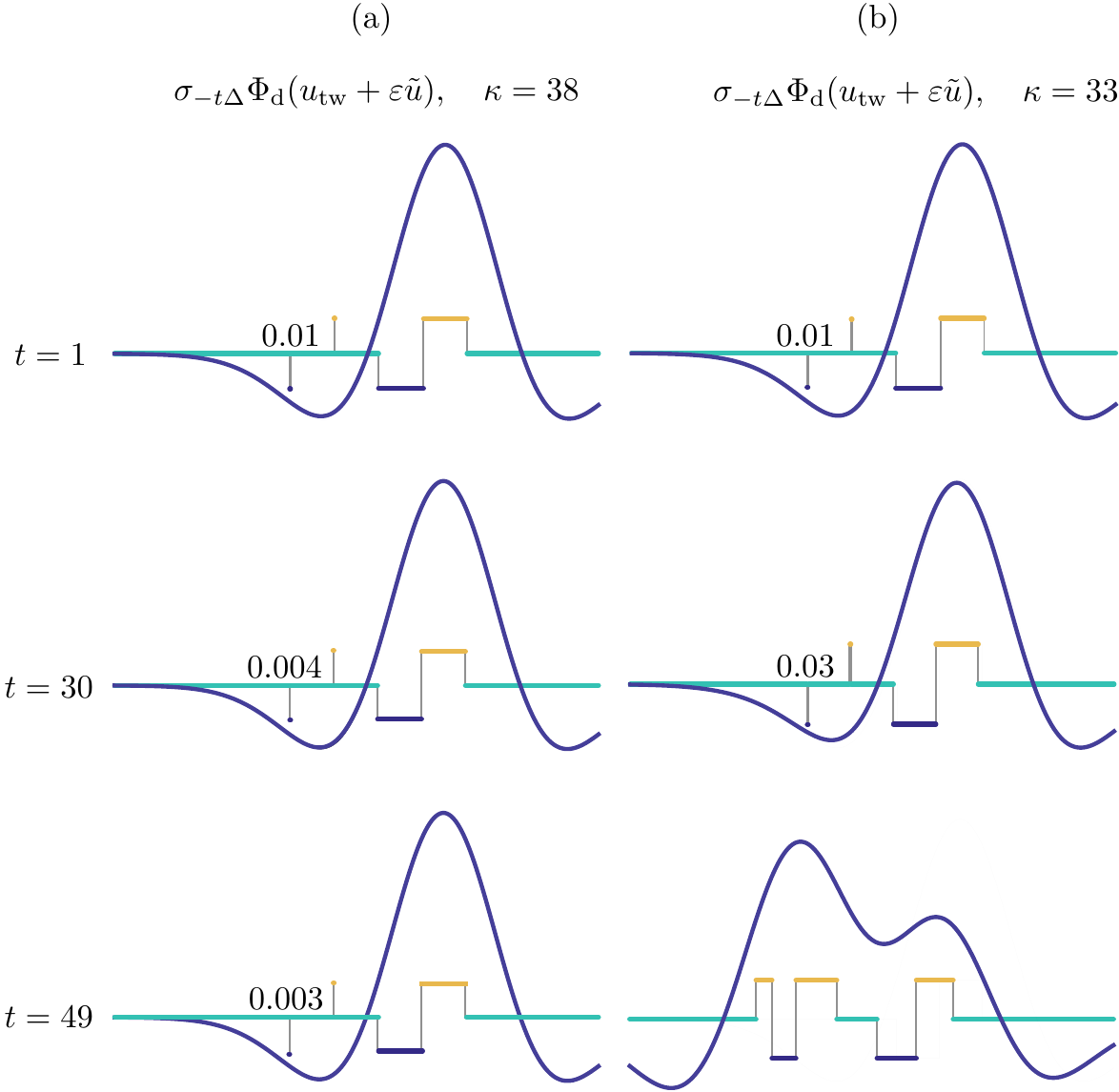}
  \caption{Numerical investigation of the linear stability of the travelling wave of
  the deterministic system, subject to perturbations in the wake of the wave. We
  iterate the map $\Phi_\mathrm{d}$ starting from a perturbed state $u_\mathrm{tw} + \eps
  \tilde u$, where $u_\mathrm{tw}$ is the mesoscopic wave profile of
  Proposition~\ref{prop:TW}, travelling with speed $\Delta$, and $\eps \tilde u$ is
  non-zero only in two intervals of width $0.01$ in the wake of the wave. We plot
  $\sigma_{-t\Delta} \Phi_\mathrm{d}( u_\mathrm{tw} + \eps \tilde u)$ and
  the corresponding macroscopic profile as a function of $t$ and we annotate the width of
  one of the perturbations. (a): For $\kappa=38$, the wave is stable. (b):
  for sufficiently small $\kappa$, the  wave becomes unstable.}
  \label{fig:waveInstability}
\end{figure}
Travelling waves in the deterministic model can also be studied via threshold
crossings, and we perform this study in the present section. We seek a measurable
function $u_\textrm{tw} \colon \SSet \to \USet$ and a constant $c
\in
\RSet$ such that 
\begin{equation}\label{eq:TWDef}
  u(x,t)= u_\textrm{tw}(x-ct) = \sum_{k \in \USet} k \ind_{X^\mathrm{tw}_k}(x-ct)
\end{equation}
almost everywhere in $\SSet$ and for all $t \in \ZSet$. We recall that, in general, a
state $u(x,t)$ is completely defined by its partition, $\{ X^\mathrm{tw}_k(t) \}$.
Consequently, Equation~\eqref{eq:TWDef} expresses that a travelling wave has a
fixed profile $u_\textrm{tw}$, whose partition, $\{ X^\mathrm{tw}_k \}$, does not depend on time.
A travelling wave $(u_\textrm{tw},c)$ satisfies almost everywhere the condition
\[
u_\textrm{tw} = \sigma_{-c} \Phidet(u_\textrm{tw};\gamma),
\]
where $\Phidet$ is the deterministic evolution operator~\eqref{eq:detPhi} and the
shift operator is defined by $\sigma_{x} \colon u(\blank) \mapsto
u(\blank-x)$.
The existence of a travelling wave is now an immediate consequence of the symmetries
of $W$\!, as shown in the following proposition. An important difference with respect
to the bump is that analytical expressions can be found for both microscopic and
mesoscopic profiles, as opposed to Proposition~\ref{prop:macroBump}, which concerns
only the mesoscopic profile.

\begin{prop}[Travelling wave] \label{prop:TW}
  Let $h, \kappa \in \RSet_+$. If there exists $\Delta \in (0,L)$ such that
  $  h = \kappa \int_{\Delta}^{2 \Delta} W(y)\, dy
  $, then
  \[
  u_\mathrm{tw}(z) = \sum_{k \in \USet} k \ind_{X^\mathrm{tw}_k}(z),
  \qquad
  \textrm{with partition}
  \qquad
  \begin{aligned}
    X^\mathrm{tw}_{-1} & = [-2\Delta,-\Delta), \\
    X^\mathrm{tw}_{0}  & = [-L,-2\Delta)\cup[0,L), \\
    X^\mathrm{tw}_{1}  & = [-\Delta,0),
  \end{aligned}
  \]
  is a travelling wave of the deterministic model~\eqref{eq:detPhi} with speed
  $c=\Delta$, associated mesoscopic profile $J_\mathrm{tw}(z) = \kappa
  \int_{-\Delta}^0 W(z-y)\, \d y$ and activity set $X^\mathrm{tw}_\geq= [-2\Delta,\Delta]$.
\end{prop}
\begin{proof} The assertion can be verified directly. We have
  \[
  \frac{h}{\kappa} = \int_\Delta^{2\Delta} w(y) dy =  
  \int_{-\Delta}^{0} W(\Delta - y) dy = 
  \int_{-\Delta}^{0} W(-2\Delta - y) dy,
  \]
  hence the activity set for $u_\textrm{tw}$ is $X^\mathrm{tw}_\geq = [-2\Delta,\Delta]$ with mesoscopic
  profile $\kappa \int_{-\Delta}^0 W(z-y)\, \d y$. Consequently,
  $\Phidet(u_\textrm{tw};\gamma)$ has partition
  \begin{align*}
    Y_{-1} & = [-\Delta,0), \\
    Y_{0}  & = [-L,-\Delta)\cup[\Delta,L), \\
    Y_{1}  & = [0,\Delta],
  \end{align*}
  and $u_\textrm{tw} = \sigma_{-\Delta}\Phidet(u_\textrm{tw},\gamma)$ almost everywhere.
\end{proof}

Numerical simulations of the deterministic model confirm the existence of the
mesoscopic travelling wave $u_\textrm{tw}$ in a suitable region of parameter space, as will
be shown in Section~\ref{sec:numerics}. The main difference between
$u_\textrm{tw}$ and the stochastic waves observed in Figure~\ref{fig:exampleWave} is in the
wake of the wave, where the former features quiescent neurons and the latter a
mixture of quiescent and refractory neurons. 

\subsection{Travelling wave stability}
As we will show in Section~\ref{sec:numerics}, waves can be found for sufficiently large
values of the gain parameter $\kappa$. However, when this parameter is below a critical value,
we observe that waves destabilise at their tail. This type of instability is
presented in the numerical experiment of Figure~\ref{fig:waveInstability}.
Here, we iterate the
dynamical system
\begin{equation}\label{eq:TWEvolution}
u(z,t+1) = \sigma_{-\Delta} \Phidet(u(z,t)), \qquad
u(z,0) = u_\textrm{tw}(z) + \eps \widetilde{u}(z),
\end{equation}
where $u_\textrm{tw}$ is the profile of Proposition~\ref{prop:TW}, travelling with
speed $\Delta$, and the perturbation $\eps \widetilde{u}_\textrm{tw}$ is non-zero only in two
intervals of width $0.01$. We deem the travelling wave stable if $u(z,t) \to
u_\textrm{tw}(z)$ as $t \to \infty$. For $\kappa$ sufficiently large, the perturbations
decay, as witnessed by their decreasing width in
Figure~\ref{fig:waveInstability}(a). For $\kappa = 33$, the perturbations grow and the
wave destabilises.
 
To analyse the behaviour of Figure~\ref{fig:waveInstability}, we shall derive the
evolution equation for a relevant class of perturbations to $u_\textrm{tw}$.
This class may be regarded as a generalisation of the perturbation applied in this
figure and is sufficient to capture the instabilities observed in numerical
simulations.
We seek solutions to~\eqref{eq:TWEvolution} with initial condition $u(z,t) = \sum_{k}
k \ind_{X_k(t)}(z)$ with time-dependent partitions
\begin{align*}
  X_{-1}(t) &= \Big[-4\Delta + \delta_1(t),-4\Delta + \delta_2(t) \Big) 
          \cup \Big[-2\Delta + \delta_5(t),- \Delta + \delta_6(t) \Big),
  \\
  \begin{split}
  X_0(t) &= \Big[-L,-4\Delta + \delta_1(t) \Big)
       \cup \Big[-4\Delta + \delta_2(t),-3\Delta + \delta_3(t) \Big)\\
    &\phantom{= \Big(-L,-4\Delta + \delta_1(t) \Big) }
    \, \cup \Big[-3\Delta + \delta_4(t),-2\Delta + \delta_5(t) \Big)
     \cup \Big[\delta_7(t), L \Big),
  \end{split} \\
  X_1(t) &=  \Big[-3\Delta + \delta_3(t),-3\Delta + \delta_4(t) \Big) 
        \cup \Big[ -\Delta + \delta_6(t),\delta_7(t) \Big),
\end{align*}
and activity set $X_\geq(t) = [\xi_1(t),\xi_2(t)]$.
In passing, we note that for $\delta_i =0$, the partition above coincides with
$\{X^\mathrm{tw}_k\}$ in Proposition~\ref{prop:TW}, hence this partition can be used as
perturbation of $u_\textrm{tw}$.
Inserting the ansatz for $u(\xi,t)$ into~\eqref{eq:TWEvolution},
we obtain a nonlinear implicit evolution equation,
$\Psi\big(\delta(t+1),\delta(t)\big)=0$, for the vector $\delta(t)$ as follows
(see Figure~\ref{fig:wavePerturbationSchematic})
\begin{align*}
& \delta_1(t+1) = \delta_3(t) \\
& \delta_2(t+1) = \delta_4(t) \\
&   \int_{-3\Delta + \delta_3(t)}^{-3\Delta + \delta_4(t)}
     w(-2\Delta + \delta_3(t+1) - y) \, dy
  + \int_{-\Delta + \delta_6(t)}^{\delta_7(t)}
     w(-2\Delta + \delta_3(t+1) - y) \, dy = h/\kappa\\
& \delta_4(t+1) = \delta_5(t) \\
& \delta_5(t+1) = \delta_6(t) \\
& \delta_6(t+1) = \delta_7(t) \\
&   \int_{-3\Delta + \delta_3(t)}^{-3\Delta + \delta_4(t)}
     w(\Delta + \delta_7(t+1) - y) \, dy
  + \int_{-\Delta + \delta_6(t)}^{\delta_7(t)}
     w(\Delta + \delta_7(t+1) - y) \, dy = h/\kappa.
\end{align*}

We note that the map above is valid under the assumption $\delta_3(t) <
\delta_4(t)$,
which preserve the number of intervals of the original partition. As
in~\cite{Kilpatrick2010bb}, we note that this prevents us from
looking at oscillatory evolution of $\delta(t)$.
We set $\delta_i(t) = \eps \lambda^t v_i$, retain terms up to first order and obtain
an eigenvalue problem for the matrix
\[
\frac{1}{\alpha}
\begin{bmatrix}
  0 & 0 & \alpha & 0 & 0 & 0 & 0 \\
  0 & 0 & 0 & \alpha & 0 & 0 & 0 \\
  0 & 0 & -w(\Delta) & w(\Delta) & 0    & -w(\Delta) & w(2\Delta) \\
  0 & 0 &         0  &         0 & \alpha &         0  & 0          \\
  0 & 0 &         0  &         0 &    0 &      \alpha  & 0          \\
  0 & 0 &         0  &         0 &    0 &         0  & \alpha       \\
  0 & 0 & w(4\Delta) &-w(4\Delta)&    0 & w(2\Delta) &-w(\Delta) \\
\end{bmatrix},
\]
where $\alpha = w(2\Delta) - w(\Delta)$.
Once again, we have an eigenvalue on the unit
circle, corresponding to a neutrally stable
translation mode. If all other eigenvalues are within the unit circle, then the wave
is linearly stable. Concrete calculations will be presented in
Section~\ref{sec:numerics}.

\begin{figure}
  \centering
  \includegraphics{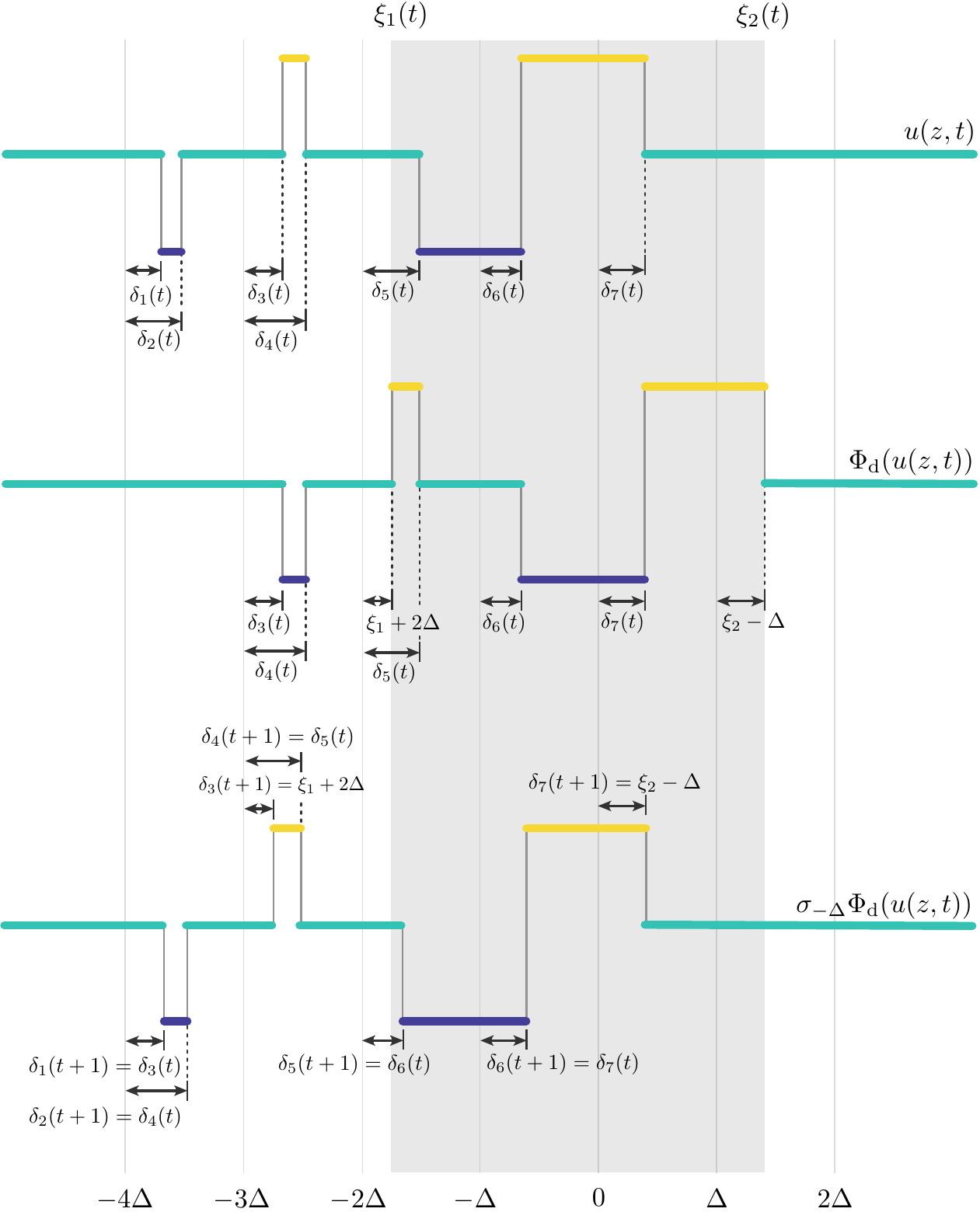}
  \caption{Visualisation of one iteration of the system~\eqref{eq:TWEvolution}: a
  perturbed travelling wave (top) is first transformed by $\Phi_\textrm{d}$
  using the rules~\eqref{eq:pullbackModel} (centre) and then shifted back by an amount
  $\Delta$ (bottom). This gives rise to an implicit evolution equation
  $\Psi\big(\delta(t+1),\delta(t)\big)=0$ for the threshold crossing points of the perturbed
  wave, as detailed in the text.}
  \label{fig:wavePerturbationSchematic}
\end{figure}


\section{Approximate probability mass functions for the Markov chain
model}\label{sec:approxMu}
We have thus far analysed coherent states of a deterministic limit of the Markov chain
model, and we now move to the more challenging stochastic setting. More precisely, we
return to the original model~\eqref{eq:microPhi} and find \textit{approximate} mass
functions for the coherent structures presented in Section~\ref{sec:simulations} (see
Figures~\ref{fig:exampleBump}--\ref{fig:exampleWave}). These approximations will be
used in the lifting procedure of the equation-free framework.

The stochastic model is a Markov chain whose $3^N$-by-$3^N$ transition kernel 
has entries specified by~\eqref{eq:kernel}. It is useful to examine the evolution
of the probability mass function for the state of a neuron at position $x_i$
in the network,  $\mu_k(x_i,t) = \prob\big( u(x_i,t) = k \big)$, $k \in \USet$,
which evolves according to
\begin{equation}\label{eq:fullMuEvolution}
\begin{bmatrix}
\mu_{-1} (x_i,t+1) \\
\mu_0 (x_i,t+1) \\
\mu_1 (x_i,t+1) \\
\end{bmatrix}=
\begin{bmatrix}
1-p & 0 & 1 \\
  p & 1-f(J(u))(x_i,t) & 0 \\
0   &   f(J(u))(x_i,t) & 0
\end{bmatrix}
\begin{bmatrix}
\mu_{-1} (x_i,t) \\
\mu_0 (x_i,t) \\
\mu_1 (x_i,t) \\
\end{bmatrix},
\end{equation}
or in compact notation $\mu(x_i,t+1) = \Pi(x_i,t) \mu(x_i,t)$. We recall that $f$ is
the sigmoidal firing rate and that $J$ is a deterministic function of the random
vector, $u(x,t) \in \USet^N$, via the
pullback set $X^u_1(t)$:
\[
  J(u)(x,t) = \kappa \int_{\XSet} W(x - y) \ind_{X^u_1(t)}(y) \, \d y.
\]
As a consequence, the evolution equation for $\mu(x_i,t)$ is non-local, 
in that $J(x_i,t)$ depends on the  microscopic state of the whole network. 

We now introduce an approximate evolution equation, obtained by posing the problem on
a continuum tissue $\SSet$ and by substituting $J(x,t)$ by its expected value
\begin{equation}\label{eq:muEvolution}
\mu(x,t+1) = \widetilde{\Pi}(x,t) \mu(x,t),
\end{equation}
where $\mu \colon \SSet \times \ZSet \to [0,1]^3$,
\begin{equation}
\widetilde{\Pi}(x,t) = 
\begin{bmatrix}
 1-p & 0 & 1 \\
  p  & 1-f\big(\mathbb{E}[J]\big)(x,t)& 0 \\
 0   &  f\big(\mathbb{E}[J]\big)(x,t) & 0
\end{bmatrix},
\end{equation}
and
\begin{equation}\label{eq:EJ}
  \mathbb{E}[J](x,t) = \kappa \int_\SSet w(x-y) \mu_1(y,t)\, \d y.
\end{equation}
In passing, we note that the evolution equation~\eqref{eq:muEvolution} is
deterministic. We are interested in two types of solutions to~\eqref{eq:muEvolution}:
\begin{enumerate}
  \item A time-independent bump solution, that is a mapping 
    $\mu_\textrm{b}$ such that $\mu(x,t)=\mu_\textrm{b}(x)$ for all $x \in \SSet$ and
    $t\in \ZSet$.
  \item A travelling wave solution, that is, a mapping
    $\mu_\textrm{tw}$ and a real number $c$ such that $\mu(x,t)=\mu_\textrm{tw}(x -
    ct)$ for all $x \in \SSet$ and $t\in \ZSet$.
\end{enumerate}

\subsection{Approximate probability mass function for bumps}\label{sec:approxMuBump}
We observe that, posing $\mu(y,t)=\mu_{\textrm{b}}(y)$
in~\eqref{eq:muEvolution}, we have
\[
\mathbb{E}[J](x) =  \kappa \int_\SSet w(x-y) (\mu_{\textrm{b}})_1(y)\, \d y , 
\]
Motivated by the simulations in Section~\ref{sec:simulations} and by
Proposition~\ref{prop:macroBump}, we seek a solution to~\eqref{eq:muEvolution} in
the limit $\beta \to \infty$, with $\mathbb{E}[J](x) \geq h$ for $x \in [0,\Delta]$,
and $(\mu_\textrm{b})_1(x) \neq 0$ for $x \in [0,\Delta]$, where $\Delta$ is unknown.
We obtain
\[
\mu_\textrm{b}(x) = \widetilde{\Pi}_\textrm{b}(x) \mu_\textrm{b}(x),
\]
where
\[
\begin{split}
  \widetilde{\Pi}_\textrm{b}(x)
  & =
    \begin{bmatrix}
     1-p & 0 & 1 \\
     p   & 1 & 0 \\
     0   & 0 & 0
    \end{bmatrix}
    \ind_{\SSet \setminus [0,\Delta]}(x)
    +
    \begin{bmatrix}
     1-p & 0 & 1 \\
     p   & 0 & 0 \\
     0   & 1 & 0
    \end{bmatrix}
    \ind_{[0,\Delta]}(x)\\
  & =
    Q_<     \ind_{\SSet \setminus [0,\Delta]}(x)
    +
    Q_\geq  \ind_{[0,\Delta]}(x),\\
\end{split}
\]
We conclude that, for each $x \in [0,\Delta]$ (respectively $x \in \SSet \setminus
[0,\Delta]$), $\mu_\textrm{b}(x)$ is the right $\Vert \blank \Vert_1$-unit
eigenvector corresponding to the eigenvalue $1$ of the stochastic matrix $Q_\geq$
(respectively $Q_<$). We find
\begin{equation}\label{eq:muBump}
\mu_\textrm{b}(x) = 
    \begin{bmatrix}
     0 \\
     1 \\
     0
    \end{bmatrix}
    \ind_{\SSet \setminus [0,\Delta]}(x)
    +
    \frac{p}{1+ 2p}
    \begin{bmatrix}
      1/p \\
      1 \\
      1
    \end{bmatrix}
    \ind_{[0,\Delta]}(x)
\end{equation}
and, by imposing the threshold condition $\mathbb{E}[J](\Delta) = h$, we obtain a
compatibility condition for $\Delta$,
\begin{equation}\label{eq:hMuBump}
h = \frac{\kappa p}{1+2p} \int_0^\Delta w(\Delta-y)\, \d y.
\end{equation}
We note that if $p=1$ we have $\mathbb{E}[J](x) = J_\textrm{b}(x,0,\Delta)$ where
$J_\textrm{b}$ is the profile for the mesoscopic bump found in
Proposition~\ref{prop:macroBump}, as expected.

\begin{figure}
  \centering
    \includegraphics{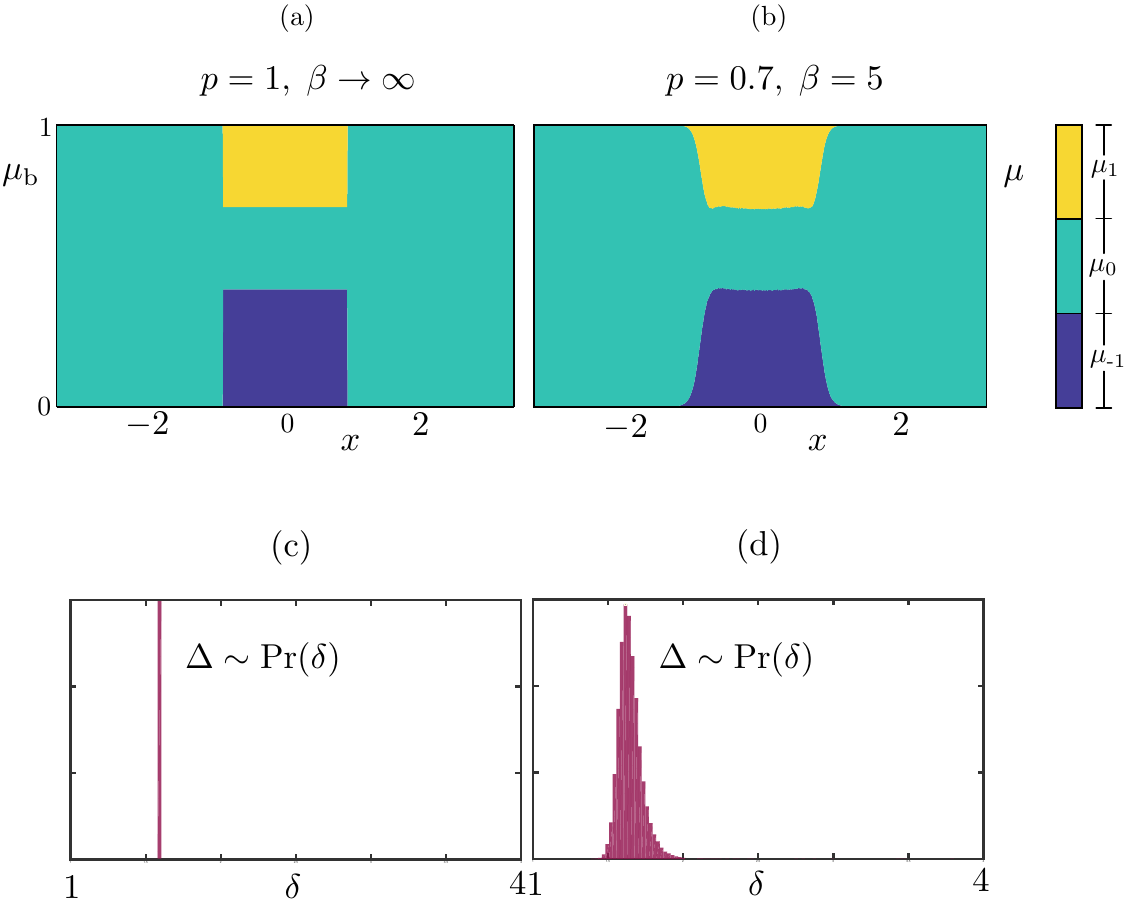}
     \caption{Comparison between the probability mass function
     $\mu_\mathrm{b}$, as computed by~\eqref{eq:muBump}-\eqref{eq:hMuBump}, and the
     observed distribution $\mu$ of the stochastic model.
     (a): We compute the vector $(\mu_\mathrm{b})_k$,  $k\in \USet$ in each strip
     using~\eqref{eq:TWdistr} and visualise the distribution using vertically
     juxtaposed color bars, with height proportional
     to the values $(\mu_\mathrm{b})_k$, as shown in the legend.
     (b): A long simulation of the stochastic model supporting a stochastic bump
     $u(x,t)$ for $t \in [0,T]$, where $T=10^5$. At each time $t>10$ (allowing
     for initial transients to decay), we compute
     $\xi_1(t)$, $\xi_2(t)$, $\Delta(t)$ and then produce histograms for the random
     profile $u(x-\xi_1(t) - \Delta(t)/2,t)$. (c): in the deterministic limit the
     value of $\Delta$ is determined by~\eqref{eq:hMuBump}, hence we have a Dirac
     distribution. (d): the distribution of $\Delta$ obtained in the Markov chain
     model. Parameters are as in Table~\ref{tab:params}.}
     \label{fig:muBumpComparison}
\end{figure}

In Figure~\ref{fig:muBumpComparison}(a), we plot $\mu_\mathrm{b}(x)$
as predicted by~\eqref{eq:muBump}--\eqref{eq:hMuBump}, for $p = 0.7,\,\kappa = 30,\,h=0.9$.
At each $x$, we visualise
$(\mu_\textrm{b})_k$ for each $k \in \USet$ using vertically juxtaposed color bars, with
height proportional to the values $(\mu_\textrm{b})_k$, as shown in the legend.
For a qualitative comparison with direct simulations, we refer the reader to the
microscopic profile $u(x,50)$ shown in the right panel of
Figure~\ref{fig:exampleBump}(a): the comparison suggests that each $u(x_i,50)$ is
distributed according to $\mu_\textrm{b}(x_i)$.

We also compared quantitatively the approximate distribution $\mu_\textrm{b}$
with the distribution, $\mu(x,t)$, obtained via Monte Carlo samples of the
full system~\eqref{eq:fullMuEvolution}. The distributions are obtained from a
long-time simulation of the stochastic model supporting a microscopic bump
$u(x,t)$ for $t \in [0,T]$, with $T=10^5$. At each discrete time
$t$, we compute the mesoscopic profile, $J(u)(x,t)$, the corresponding threshold
crossings and width: $\xi_1(t)$, $\xi_2(t)$, $\Delta(t)$ and then produce
histograms for the random profile $u(x-\xi_1(t) - \Delta(t)/2,t)$. The instantaneous
shift applied to the profile is necessary to pin the wandering bump.

We note a discrepancy between the analytically computed histograms, in which we observe
a sharp transition between the region $x\in[0,\Delta]$ and $x\in\SSet\setminus [0,\Delta]$,
and the numerically computed ones, in which this transition is smoother. This
discrepancy arises because $\Delta(t)$ oscillates around an average value $\Delta$
predicted by~\eqref{eq:hMuBump}; the approximate evolution
equation~\eqref{eq:muEvolution} does not account for these oscillations. This is
visible in the histograms of Figure~\ref{fig:muBumpComparison}(c)-(d), as well as in the
direct numerical simulation~\ref{fig:slaving}(a).

\subsection{Approximate probability mass function for travelling waves}
\label{sec:approxMuWave}
We now follow a similar strategy to approximate the probability mass function for
travelling waves. We pose $\mu(x,t) = \mu_\textrm{tw}(x-ct)$ in the expression for
$\mathbb{E}[J]$, to obtain
\[
\kappa \int_\SSet w(x-y) (\mu_{\textrm{tw}})_1(y-ct)\, \d y  =
\kappa \int_\SSet w(x-ct-y) (\mu_{\textrm{tw}})_1(y)\, \d y  =
\mathbb{E}[J](x-ct).
\]

Proposition~\ref{prop:TW} provides us with a deterministic travelling wave with
speed $c=\Delta$. The parameter $\Delta$ is also connected to the mesoscopic wave
profile, which has threshold crossings $\xi_1= -2\Delta$ and $\xi_2=\Delta$.
Hence, we seek for a solution to~\eqref{eq:muEvolution} in the limit $\beta
\to \infty$, with $\mathbb{E}[J](z) \geq h$ for $x \in [-2c,c]$,
and $(\mu_\textrm{tw})_1(z) \neq 0$ for $z \in [-2c,c]$, where $c$ is
unknown. For simplicity, we pose the problem on a large domain whose size is
commensurate with $c$, that is $\SSet =cT/\RSet$, where $T$ is an even integer
much greater than 1.

We obtain
\[
\sigma_{ct} \mu_\textrm{tw}(z)
=
  \widetilde{\Pi}_\textrm{tw}(z - c(t-1))
  \widetilde{\Pi}_\textrm{tw}(z - c(t-2))
  \cdots
  \widetilde{\Pi}_\textrm{tw}(z) \mu_\textrm{tw}(z),
\]
where
\[\
\widetilde
\Pi_\textrm{tw}(z) =
Q_<     \ind_{\SSet \setminus [-c,c]}(z)
+
Q_\geq  \ind_{[-2c,c]}(z).
\]

To make further analytical progress, it is useful to partition the domain  $\SSet
=cT/\RSet$ in strips of width $c$,
\[
  \SSet = \bigcup_{j=T/2}^{T/2} \, \big[ j c, (j+1) c \big) =
  \bigcup_{j=T/2}^{T/2-1} I_j(c),
\]
and impose that the wave returns back to its original position after $T$ iterations,
$\sigma_{cT} \mu_\textrm{tw}(z)= \mu_\textrm{tw}(z)$, while satisfying
the compatibility condition $h = \mathbb{E}[J](c)$. This leads to the system
\begin{equation}\label{eq:TWdistr}
  \begin{aligned}
  & \mu_\textrm{tw}(z)  = R(z,c) \mu_\textrm{tw}(z) = \Bigg[ \sum_{j=-T/2}^{T/2-1} R_j
  \ind_{I_j(c)}(z) \Bigg] \mu_\textrm{tw}(z), \\
   & \kappa \int_{-2c}^c W(c-y) (\mu_\textrm{tw})_1(y)\, \d y = h.
  \end{aligned}
\end{equation}
With reference to system~\eqref{eq:TWdistr} we note that:
\begin{enumerate}
  \item $R(z,c)$ is constant within each strip $I_j$, hence the probability mass function,
    $\mu_\textrm{tw}(z)$, is also constant in each strip, that is,
    $\mu_\textrm{tw}(z) = \sum_i \rho_i \ind_{I_i(c)}(z)$ for some unknown
    vector $(\rho_{-T/2},\ldots,\rho_{T/2}) \in \SSet^{3T}$.

  \item Each $R_j$ is a product of $T$ $3$-by-$3$ stochastic matrices,
    each equal to $Q_<$ or $Q_\geq$. Furthermore, the matrices $\{ R_j \}$ are computable. For
    instance, for the strip $I_{-1}$ we have
    \[
    R_{-1}=
    \begin{bmatrix}
    (1-p)^{T} + p(1-p)^{T-2}      &  (1-p)^{T-2}   &   (1-p)^{T-1}    \\[0.2em]
    p(1-p)^{T-1} + p^2(1-p)^{T-3} &  p(1-p)^{T-3}  &   p(1-p)^{T-2}   \\[0.2em] 
    1-(1-p)^{T-1}-p(1-p)^{T-3}    &  1-(1-p)^{T-3} &   1-(1-p)^{T-2}  \\
    \end{bmatrix}.
    \]
\end{enumerate}
Consequently, $\mu_\textrm{tw}(z)$ can be determined
by solving the following problem in the unknown
$(\rho_{-T/2},\ldots,\rho_{T/2},c) \in \SSet^{3T} \times \RSet$:
\begin{equation}\label{eq:findRho}
\begin{aligned}
  \rho_i - R_i \rho_i & = 0, \qquad & i = -T/2,\ldots,T/2-1, \\
  - h + \kappa (\rho_{-1})_1 \displaystyle{\int_{-c}^0} W(c-y) \, \d y & = 0. &
\end{aligned}
\end{equation}

\begin{figure}
  \centering
    \includegraphics{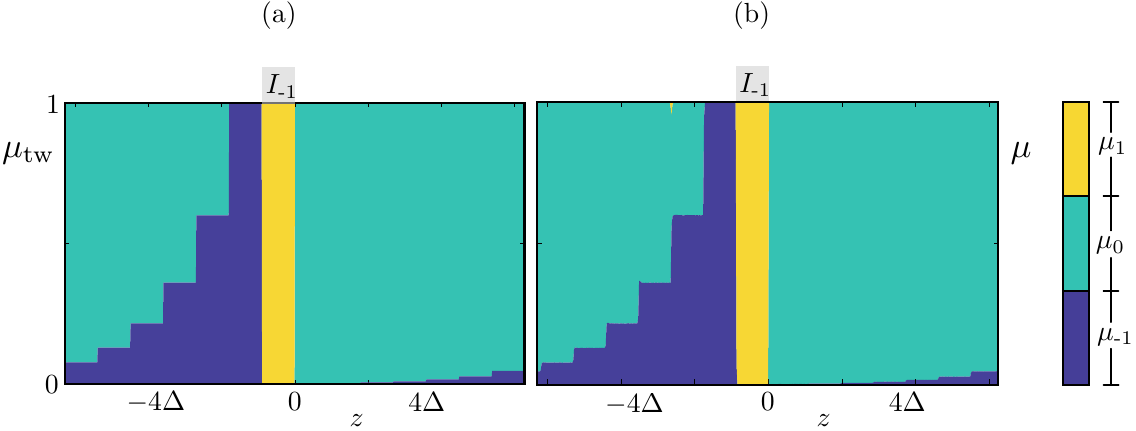}
    \caption{Similarly to Figure~\ref{fig:muBumpComparison}, we compare the
    approximated probability mass function $\mu_\textrm{tw}$, and the observed
    distribution $\mu$ of the stochastic model.
    (a): the probability mass function is approximated using the numerical scheme
    outlined in the main text for the solution of~\eqref{eq:findRho}; the strip
    $I_{-1}$ is indicated for reference. (b): A set of $9\times 10^5$ realisations of the
    stochastic model for a travelling wave are run for $t \in
    [0,T]$, where $T=1000$. For each realisation $s$, we calculate the final threshold
    crossings $\xi^s_1(T)$, $\xi^s_2(T)$, and then compute histograms of
    $u^s(x-\xi^s_2(T),T)$. We stress that the strips in (a) are induced by our
    numerical procedure, while the ones in (b) emerge from the data. The
    agreement is excellent and is preserved across a vast region of parameter
    space (not shown). Parameters are as in Table~\ref{tab:params}.}
    \label{fig:muComparison}
\end{figure}

Before presenting a quantitative comparison between the numerically determined
distribution,
$\mu_\textrm{tw}(z)$, and that obtained via direct time simulations, we
make a few efficiency considerations. In the following sections, it will become
apparent that sampling the distribution $\mu_\textrm{tw}(z)$ for various values of
control parameters, such as $h$ or $\kappa$, is a recurrent task, at the core of the
coarse bifurcation analysis: each linear and nonlinear function evaluation within the
continuation algorithm requires sampling $\mu_\textrm{tw}(z)$, and hence solving the
large nonlinear problem~\eqref{eq:findRho}.

With little effort, however, we can obtain an accurate \textit{approximation} to
$\mu_\textrm{tw}$, with considerable computational savings. The inspiration comes
once again from the analytical wave of Proposition~\ref{prop:TW}. We notice that only
the last equation of system~\eqref{eq:findRho} is nonlinear; the last equation is
also the only one which couples $\{\rho_j\}$ with $c$. When
$p=1$ the wave speed is known as $\beta \to \infty$, $N\to \infty$ and $p=1$
corresponds to the deterministic limit, hence
$\mathbb{E}[J](z) = J_\textrm{tw}(z)$, which implies 
$c = \Delta$ and $(\rho_{-1})_1 = 1$. The stochastic waves observed in direct
simulations for $p\neq 1$, however, display $c \approx (\xi_2 - \xi_1)/3 = \Delta$
and $\mu \approx 1$ in the strip where $J$ achieves a local maximum (see, for
instance Figure~\ref{fig:exampleWave}, for which $p = 0.4$).

The considerations above lead us to the following scheme to approximate
$\mu_\textrm{tw}$: (i) set $c = \Delta$ and remove the last equation
in~\eqref{eq:findRho};  (ii) solve $T$ decoupled $3$-by-$3$ eigenvalue problems to
find $\rho_i$. Furthermore, if $p$ remains fixed in the coarse bifurcation analysis,
$\rho_i$ can be pre-computed and step (ii) can be skipped.

In Figure~\ref{fig:muComparison}(a), we report the approximate $\mu_\textrm{tw}$
found with the numerical procedure described above. An inspection of the microscopic
profile $u(x,45)$ in the right panel of Figure~\ref{fig:exampleWave}(a) shows that
this profile is compatible with $\mu_\textrm{tw}$. We also compared quantitatively
the approximate distribution with the distribution, $\mu(x,t)$, obtained with
Monte Carlo samples of the full system~\eqref{eq:fullMuEvolution}. The distributions
are obtained from $M$ samples $\{u^s(x,t)\}^M_{s=1}$ of the stochastic model for a
travelling wave for $t \in [0,T]$. For each sample $s$, we compute the thresholds,
$\xi^s_1(T)$, $\xi^s_2(T)$, of the corresponding $J(u^s)(x,T)$ and then produce
histograms for $u^s(x-\xi^s_2(T),T)$. This shifting, whose results are reported in
Figure~\ref{fig:muComparison}(b), does not enforce any constant value for the
velocity, hence it allows us to test the numerical approximation
$\mu_\textrm{tw}$. The agreement between the two distributions is excellent: we stress that, while the strips in
Figure~\ref{fig:muComparison}(a) are enforced by our approximation, the ones in
Figure~\ref{fig:muComparison}(b) emerge from the data. We note a slight discrepancy,
in that $\mu_\textrm{tw}(-3\Delta) \approx 0$, while the other distribution shows
a small nonzero probability attributed to the firing state at $\xi=-3\Delta$. Despite
this minor disagreement, the differences between the approximated and observed
distributions remain small across all parameter regimes of note and the
approximations even retain their accuracy as $\beta$ is decreased (not
shown).

\section{Coarse time-stepper}\label{sec:coarseTimeStepper}
As mentioned in the introduction, equation-free methods allow us to compute macroscopic
states in cases in which a macroscopic evolution equation is not available in closed
form~\cite{Kevrekidis2003r,Kevrekidis:2009jo}. To understand the general idea behind
the equation-free framework, we initially discuss an example taken from one of the
previous sections, where an evolution equation \textit{does} exist in closed form. 

In Section~\ref{sec:MacroBump}, we described bumps in a deterministic limit of the
Markov chain model. In this description, we singled out a \textit{microscopic} state
(the function $u_m(x)$ with partition~\eqref{eq:Xm}) and a
corresponding \textit{mesoscopic} state (the function $J_m(x)$), both sketched in
Figure~\ref{fig:analyticBump}. Proposition~\ref{prop:macroBump} shows that there
exists a well defined mesoscopic limit profile, $J_\textrm{b}$, which is determined (up to
translations in $x$) by its threshold crossings $\xi_1=0$, $\xi_2=\Delta$. This
suggests a characterisation of the bump in terms of the \textit{macroscopic} vector
$(\xi_1,\xi_2)$ or, once translation invariance is factored out, in terms of the
\textit{macroscopic} bump width, $\Delta$. Even though the microscopic state $u_m$ is not
an equilibrium of the deterministic system, the macroscopic state $(0,\Delta)$ is a
fixed point of the evolution equation~\eqref{eq:phiBump}, whose evolution operator
$\Psi$ is known in closed form, owing to Proposition~\ref{prop:macroBump}. It is then
possible to compute $\Delta$ as a root of an explicitly available nonlinear equation.

We now aim to use equation-free methods to compute macroscopic equilibria in cases
where we do not have an explicit evolution equation, but only a numerical procedure
to approximate $\Psi$. As mentioned in the introduction, the evolution equation is
approximated using a coarse time-stepper, which maps the macroscopic state at
time $t_0$ to the macroscopic state at time $t_1$ using three stages: lifting,
evolution, restriction. The specification of these stages (the lifting in particular)
typically requires some closure assumptions, which are enforced numerically. In our
case, we use the analysis of the previous sections for this purpose. In the following
section, we discuss the coarse time-stepper for bumps and travelling waves. The multi-bump
case is a straightforward extension of the single bump case.

\subsection{Coarse time-stepper for bumps}\label{sec:bumpLift}

The macroscopic variables for the bump are the threshold crossings $\{\xi_i\}$ of the
mesoscopic profile $J$. The lifting operator for the bump takes as arguments
$\{ \xi_i \}$ and returns a set of microscopic profiles compatible with these
threshold crossings: 
\[
\mathcal{L}_\textrm{b} \colon \SSet^2 \to \USet^{N \times M},
\qquad
(\xi_1,\xi_2)^\mathrm{T} \mapsto
\{ u^s(x) \}_s.
\]

\begin{figure}
  \centering
    \includegraphics{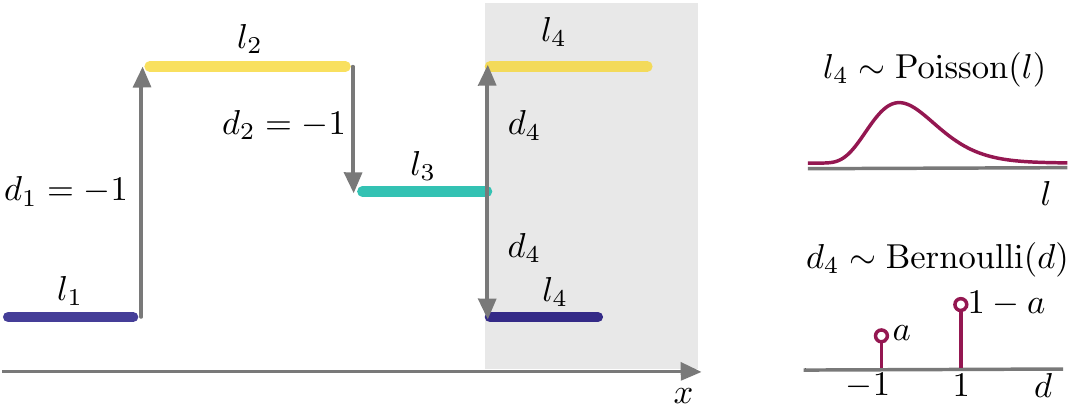}
    \caption{\label{fig:PoissonBernoulliLifting} Schematic representation of the lift
    operator for a bump solution. This figure displays a representation of how
    the states for neurons located within the activity set, $[\xi_1,\xi_2]$, are lifted.
    For illustrative purposes, we assume here that we are midway through the lifting
    operation, where 3 steps of the \emph{while} loop listed in Algorithm~\ref{alg:bumpLift}
    have been completed and a fourth one is being executed (shaded area). The width
    $l_4$ of the next strip is drawn from a Poisson
    distribution. The random variable $d\in\lbrace -1, 1\rbrace$ indicates the
    direction through which we cycle through the states $\lbrace -1, 0, 1 \rbrace$
    during the lifting. The number $d_4$ is drawn from a
    Bernoulli distribution whose average $a$ gives the probablity of changing
    direction. For full details of the lifting operator, please refer to
    Algorithm~\ref{alg:bumpLift}.}
\end{figure}

If $\beta \to \infty$, $u^s(x)$ are samples of the analytical
probability mass function $\mu_\textrm{b}(x + \Delta/2)$, where $\mu_\textrm{b}$ is
given by \eqref{eq:muBump} with $\Delta =\xi_2-\xi_1$. In this limit, a solution
branch may also be traced by plotting~\eqref{eq:hMuBump}. 

If $\beta$ is finite, we either extract samples from the \textit{approximate}
probability mass function $\mu_\textrm{b}$ used above, or we extract samples $u^s(x)$
satisfying the following properties (see Proposition~\ref{prop:macroBump} and
Remark~\ref{rem:perm}):
\begin{enumerate}
  \item $u^s(x)$ is symmetric with respect to the axis $x = (\widetilde{\xi}_1 +
    \widetilde{\xi}_2)/2$, where $\tilde \xi_i = \round(\xi_i)$ and $\round \colon \SSet
    \to \SnSet$.
  \item $u^s(x) = 0$ for all $x \in [-L,\widetilde{\xi}_1) \cup (\widetilde{\xi}_2,L)$.
    \item The pullback sets, $X_1$ and $X_{-1}$, are contained within $[\widetilde{\xi}_1,
      \widetilde{\xi}_2]$ and are unions of a random number of intervals whose
      widths are also random. A schematic of the lifting operator for bumps is
      shown in Figure~\ref{fig:PoissonBernoulliLifting}.
\end{enumerate}
A more precise description of the latter sampling is given in
Algorithm~\ref{alg:bumpLift}. As mentioned in the introduction, lifting operators are
not unique and we have given above two possible examples of lifting.
In our computations, we favour the second sampling method.
The mesoscopic profiles, $J$, generated using this approach are well-matched
to $\mathbb{E}[J]$ produced by the analytically derived probability mass
functions \eqref{eq:muBump}.
Numerical experiments demonstrate that this method is better than the first
possible lifting choice at continuing unstable branches.
This is most likely due to the fact that the latter method slightly
overestimates the probability of neurons within the bump to be in the
spiking state, and underestimates that of them being in the refractory state
and this helps mitigate the problems encountered when finding unstable states
caused by the combination of the
finite size of the network and non-smooth characteristics of the model
(when $\beta$ is high).

\begin{algorithm}
  \caption{\label{alg:bumpLift}Lifting operator for bump}
  \SetAlgoNoLine
  \DontPrintSemicolon
   \SetKwInOut{Input}{Input}
   \SetKwInOut{Output}{Output}
   \SetKwInOut{Suggestions}{Suggestions}
   \SetKwInOut{Comments}{Comments}
 
   \Input{Threshold crossings $\xi_1,\xi_2$; Average number of strips, $m$; Average
   for the Bernoulli distribution, $a$; Width of the domain, $L$.}
   \Output{Profiles $u^1(x),\ldots,u^M(x)$}
  \Comments{The profiles $u^s(x)$ are assumed to be symmetric around $x=(\xi_1 +
    \xi_2)/2$. The operator \emph{round} rounds a real number to a computational grid with
  stepsize $\delta x = 2L/N$.
  }
   \textbf{Pseudocode}:\;
    \For{$s = 1,M$}{
      Set $u^s(x) = 0$ for all $x \in [-L,\xi_1) \cup (\xi_2,L)$\;
      Set $d=-1$\;
      $x = \textrm{round}(\xi_1)$, $u^s(x) = 1$\;
      \While(){$x \leq (\xi_1 + \xi_2)/2$}{
	Select random width $l \sim \textrm{Poisson}( (\xi_2-\xi_1)/m)$\;
	Select random increment $b \sim \textrm{Bernoulli}(a)$,
  $d = \left(d + 2b + 1\right) \!\!\mod 3 - 1$\;
	\For{$j = 1,l$}{
	  Update $x = x + \delta x$\;
	  \If{$x \leq (\xi_1 + \xi_2)/2$}{
	    \eIf($u^i$ changes value at the next grid point) {$j=l$}{
	      $u^s(x+\delta x) =  (u^s(x) + d + 1) \!  \mod 3 - 1$\;
	    }($u^s$ remains constant at the next grid point){
	      $u^s(x+\delta x) = u^s(x)$\;
	    }
	    Reflection around symmetry axis, $u(\xi_2 + \xi_1-x) = u(x)$\;
	  }
        }
    }
  }
\end{algorithm}

The evolution operator is given by
\[
\mathcal{E}_T \colon \USet^{N \times M} \to \USet^{N \times M},
\qquad
\{ u^j(x) \}_j \mapsto \{ \varphi_T( u^j(x)) \}_j,
\]
where $\varphi_T$ denotes $T$ compositions of the microscopic evolution
operator~\eqref{eq:microPhi} and the dependence on the control
parameter, $\gamma$, is omitted for simplicity.

For the restriction operator, we compute the average activity set of the profiles.
More specifically, we set
\[
\mathcal{R} \colon  \USet^{N \times M} \to \SSet^2,
\qquad
\{ u^j(x) \}_j \mapsto ( \xi_1, \xi_2 )^\mathrm{T},
\]
where
\[
  \xi_i = \frac{1}{M} \sum_{s=1}^M \xi_i^s, \quad i = 1,2,
\]
and $\xi^s_i$ are defined using a piecewise first-order interpolant $\mathcal{P}^3_N J$
of $J$ with nodes $\{x_i\}_{i=0}^N$,
\begin{align*}
& \xi^s_1 = 
\bigg\{ x \in \SSet \colon \quad \mathcal{P}^3_N J(u^s)(x) = h, \quad  \frac{\d}{\d x} \mathcal{P}^3_N J(u^s)(x) > 0 \bigg\},
\\
& \xi^s_2 = 
\bigg\{ x \in \SSet \colon \quad \mathcal{P}^3_N J(u^s)(x) = h, \quad  \frac{\d}{\d x}
\mathcal{P}^3_N J(u^s)(x) < 0 \bigg\}.
\end{align*}
We also point out that the computation stops if the two sets above are
empty, whereupon, we set $\xi_1^s=\xi_2^s=0$.

The coarse time-stepper for bumps is then given by 
\begin{equation}\label{eq:coarseTimeStepperBump}
\Phi_\textrm{b} \colon \SSet^2 \to \SSet^2,
\qquad
\xi \mapsto ( \mathcal{R} \circ \mathcal{E}_T \circ \mathcal{L}_\textrm{b} )
(\xi),
\end{equation}
where the dependence on parameter $\gamma$ has been omitted.

\subsection{Coarse time-stepper for travelling waves}\label{sec:TWpLift}
In Section~\ref{sec:approxMuBump}, we showed that the probability mass function,
$\mu_\textrm{tw}(z)$, of a coarse travelling wave can be approximated numerically
using the travelling wave of the deterministic model, by solving a simple set of
eigenvalue problems. It is therefore natural to use $\mu_\textrm{tw}$ in the lifting
procedure for the travelling wave. In analogy with what was done for the
bump, our coarse variables $(\xi_1,\xi_2)$ are the boundaries of the activity set
associated with the coarse wave, $X_\geq = [\xi_1, \xi_2]$. We then set
\[
\mathcal{L}_\textrm{tw} \colon \XSet^2 \to \USet^{N \times M},
\qquad
(\xi_1,\xi_2)^T \mapsto
\{ u^s(x) \}_s,
\]
where $\{u^s(x_i)\}_s$ are $M$ independent samples of the probability mass functions
$\mu_\textrm{tw}(x_i)$, with $c=(\xi_2-\xi_1)/3$. The restriction operator for
travelling waves is the same used for the bump. The coarse time-stepper for travelling
waves, $\Phi_\textrm{tw}$, is then obtained as in~\eqref{eq:coarseTimeStepperBump},
with $\mathcal{L}_\textrm{b}$ replaced by $\mathcal{L}_\textrm{tw}$.
\begin{equation}\label{eq:coarseTimeStepperWave}
\Phi_\textrm{tw} \colon \SSet^2 \to \SSet^2,
\qquad
\xi \mapsto ( \mathcal{R} \circ \mathcal{E}_T \circ \mathcal{L}_\textrm{tw} )
(\xi).
\end{equation}

\section{Root finding and pseudo-arclength continuation}
\label{sec:continuation}
Once the coarse time-steppers, $\Phi_\textrm{b}$ and $\Phi_\textrm{tw}$, have been
defined, it is possible to
use Newton's method and pseudo-arclength continuation to compute coarse states,
continue them in one of the control parameters and assess their coarse linear
stability. In this section, we will indicate dependence upon a single parameter
$\gamma \in \RSet$, implying that this can be any of the control parameters
in~\eqref{eq:microPhi}. 

For bumps, we continue in $\gamma$ the nonlinear problem
$F_\textrm{b}(\xi;\gamma)=0$, where
\begin{equation}\label{eq:Fb}
F_\textrm{b} \colon \SSet^2 \times \RSet \to \SSet^2,
\qquad
\xi \mapsto 
\begin{bmatrix}
  \xi_1 \\
  \xi_2 - \big( \Phi_\textrm{b}(\xi; \gamma) \big )_2
\end{bmatrix}.
\end{equation}
A vector $\xi$ such that $F_\textrm{b}(\xi;\gamma)=0$ corresponds to a coarse bump
with activity set $X_\geq = [0,\xi_2]$ and width $\xi_2$, occurring for the parameter
value $\gamma$, that is, we eliminated the translation invariance associated with the
problem by imposing $\xi_1=0$. In passing, we note that it is possible to hardwire
the condition $\xi_1=0$ directly in $F_\textrm{b}$ and proceed to solve an equivalent
1-dimensional system. Here, we retain the 2-dimensional formulation with the explicit
condition $\xi_1=0$, as this makes the exposition simpler.

During continuation, the explicitly unavailable Jacobians
\[
D_\xi F_\textrm{b}(\xi;\gamma) = I -  D_\xi \Phi_\textrm{b}(\xi; \gamma),
\qquad
D_\gamma F_\textrm{b}(\xi;\gamma) = D_\gamma \Phi_\textrm{b}(\xi; \gamma),
\]
are approximated using the first-order forward finite-difference formulas
\begin{align*}
  & \eps D_\xi \Phi_\textrm{b}(\xi; \gamma) \widetilde{\xi} \approx 
  \Phi_\textrm{b}(\xi + \eps \widetilde{\xi}; \gamma) - \Phi_\textrm{b}(\xi; \gamma),  \\
& \eps D_\gamma \Phi_\textrm{b}(\xi; \gamma) \widetilde{\gamma} \approx
\Phi_\textrm{b}(\xi; \gamma + \eps \widetilde{\gamma}) - \Phi_\textrm{b}(\xi; \gamma).
\end{align*}

%
The finite difference formula for $D_\xi \Phi_\textrm{b}$ also defines the Jacobian
operator used to compute stability: for a given solution $\xi_*$ of~\eqref{eq:Fb}, we
study the associated eigenvalue problem
\[
\lambda \xi = D_\xi \Phi_\textrm{b}(\xi_*; \gamma) \xi, \qquad \lambda \in \CSet,
\quad \xi \in \RSet^2.
\]

For coarse travelling waves, we define
\begin{equation}\label{eq:Ftw}
F_\textrm{tw} \colon \SSet^2 \times \RSet^2 \to \SSet^2,
\qquad
\begin{bmatrix}
  \xi \\
  c
\end{bmatrix}
\mapsto
\begin{bmatrix}
   \xi_1\\
   \xi_2    -cT + \big(\Phi_\textrm{b}\big(\xi; \gamma\big) \big)_2 \\
   c - (\xi_2-\xi_1)/3
\end{bmatrix}.
\end{equation}
A solution $(\xi,c)$ to the problem $F_\mathrm{tw}(\xi,c;\gamma) = 0$ corresponds
to a coarse travelling wave with activity set $X_\geq = [0,\xi_2]$ and speed
$\xi_2/3$, that is, we eliminated the translation invariance and imposed a speed
$c$ in accordance with the lifting procedure $\mathcal{L}_\mathrm{tw}$. As for
the bump we can, in
principle, solve an equivalent 1-dimensional coarse problem.

\section{Numerical results}\label{sec:numerics}
\begin{figure}
  \centering
  \includegraphics{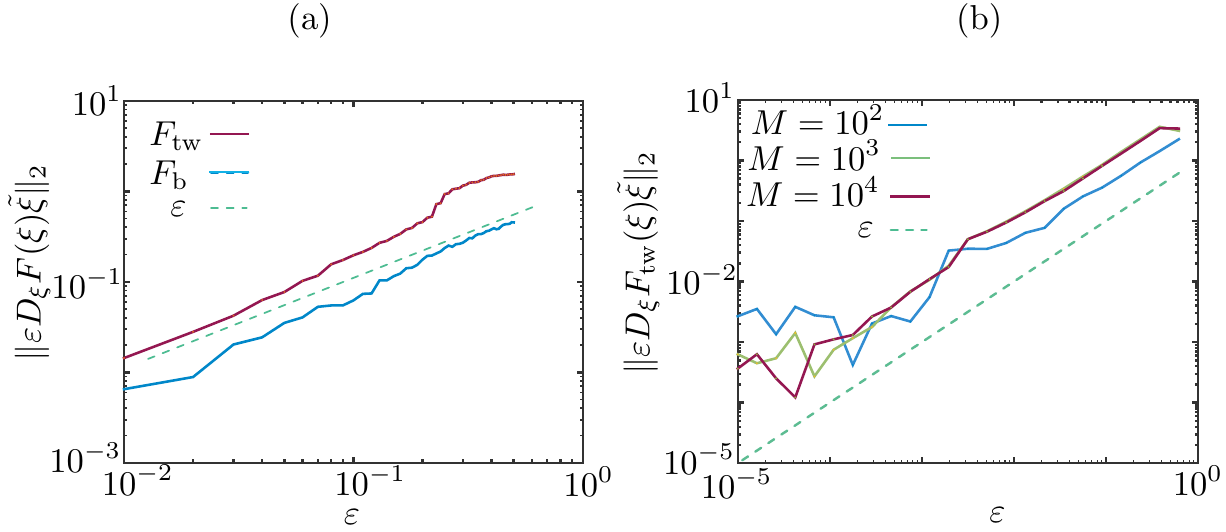}
  \caption{Jacobian-vector product norm as a function of $\eps$. The approximated
    Jacobian-vector products $D_\xi F_\mathrm{b}(\xi) \widetilde{\xi}$ and $D_\xi
    F_\mathrm{tw}(\xi) \widetilde{\xi}$, are evaluated at a coarse bump and a coarse
  travelling wave $\xi$ in a randomly selected direction $\eps \widetilde{\xi}$, where $\Vert \xi
  \Vert_2 = 1$. (a) A single realisation of the deterministic coarse-evolution maps
  is used in the test, showing that the norm of the Jacobian-vector product is
  an $\mathcal{O}(\eps)$, as expected. Parameters: $p=1$, $\kappa=30$, $\beta \to
  \infty$ (Heaviside firing rate), $h=1$, $N=128$, $A_1 = 5.25$, $A_2=5$, $B_1=0.2$,
  $B_2=0.3$. (b) The experiment is repeated for a coarse travelling wave in the
  stochastic setting and for various values of $M$. Parameters as in (a), except
  $p=0.4$.}
  \label{fig:jacobianVersusEpsi}
\end{figure}
\begin{figure}
  \centering
  \includegraphics{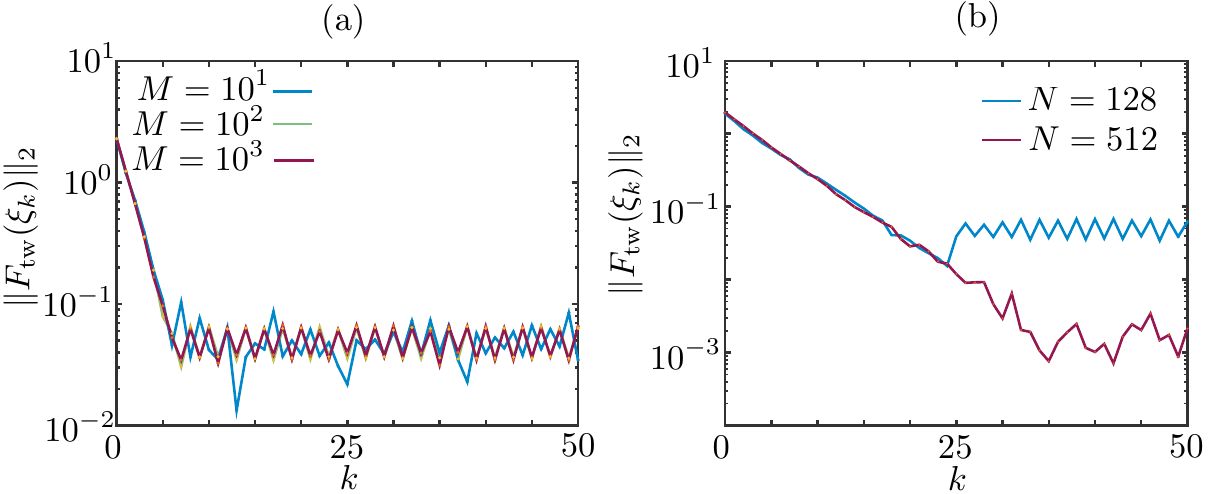}
  \caption{Convergence history of the damped Newton's method applied to the coarse
  travelling wave problem. (a): the method converges linearly, and the achievable
  tolerance does not decrease when the number of realisations $M$ is increased. 
  (b): the achievable tolerance depends on the grid size, or, equivalently, on the
  number of neurons, $N$.}
  \label{fig:newtonConvergence}
\end{figure}
We begin by testing the numerical properties of the coarse time-stepper, the
Jacobian-vector products and the Newton solver used for our computations. 
In Figure~\ref{fig:jacobianVersusEpsi}(a), we evaluate the Jacobian-
vector product of the coarse time stepper with $p=1$, $\beta \to \infty$ for bumps
(waves) evaluated at a coarse bump (wave), in the direction $\eps \widetilde{
\xi}$, where $0 < \eps \ll 1$ and $\widetilde{\xi}$ is a random vector with norm 1. Since
this coarse time stepper corresponds to the deterministic case, we expect the norm of
the Jacobian-vector product to be an $O(\eps)$, as confirmed by the numerical
experiment. In Figure~\ref{fig:jacobianVersusEpsi}(b), we repeat the experiment in
the stochastic setting ($p=0.4$), for the travelling wave case with various number of
realisations.  As expected, the norm of the Jacobian-vector action follows the
$O(\eps)$ curve for sufficiently large $\eps$: the more realisations are employed,
the more accurately the $O(\eps)$ curve is followed. 

\begin{figure}
  \centering
    \includegraphics{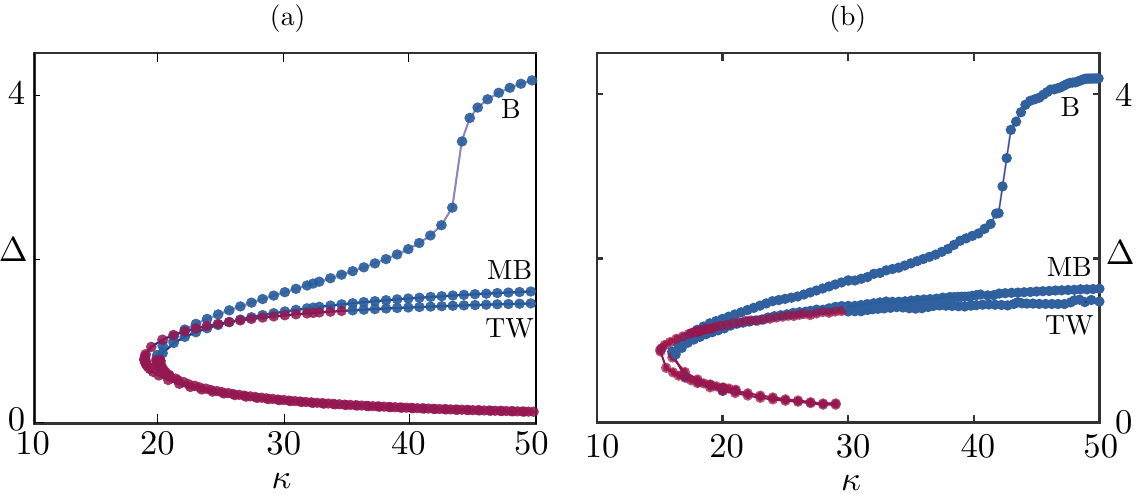}
  \caption{Bifurcation diagrams for bumps (B), multibumps (MB) and travelling
  waves (TW) using $\kappa$ as bifurcation parameter parameter. (a):
  Using the analytical results, we see that bump, multi-bump and travelling
  wave solutions coexist and are stable for sufficiently high $\kappa$ (see main
  text for details). (b): The solution branches found using the
  equation-free methods agree with the analytical results. Parameters as in
  Table~\ref{tab:params} except $h=p=1.0$, $\beta\to\infty$.}
  \label{fig:bifDiagComparison}
\end{figure}

\begin{figure}
  \centering
  \includegraphics[width=4in]{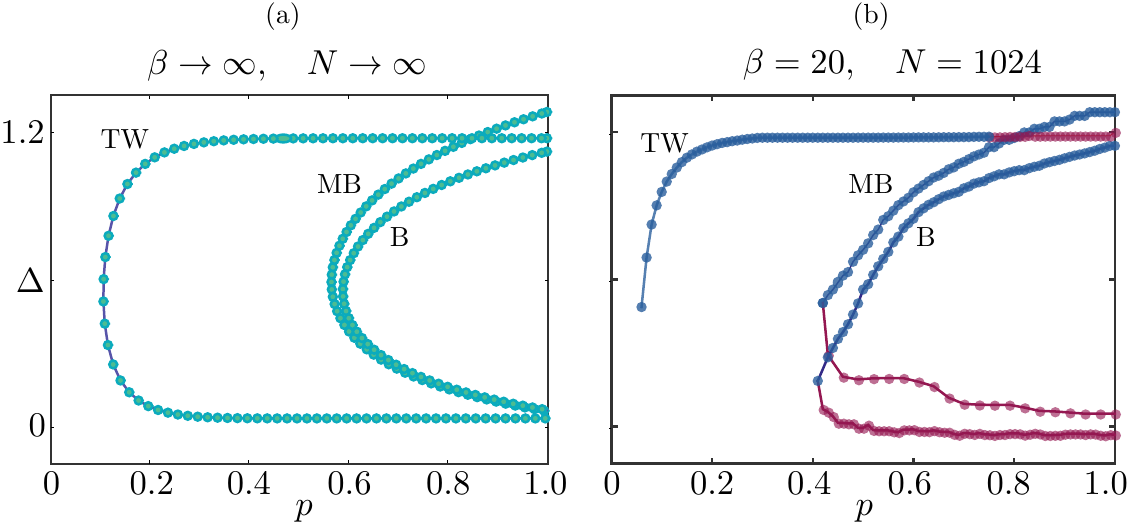}
    \caption{Bifurcation in the control parameter $p$. (a): Existence curves
    obtained analytically; we see that, below a critical value of $p$, only the
    travelling wave exists. (b): The solution branches found using the
    equation-free method agree qualitatively with the analytical results, and we
    can use the method to infer stability. For full details, please refer to the
    text. Parameters as in Table~\ref{tab:params} except $\kappa=20.0,\,h=0.9$,
    with $\beta\to\infty$ for (a) and $\beta=20.0$ for (b).}
  \label{fig:continuationP}
\end{figure}

\begin{figure}
  \centering
    \includegraphics{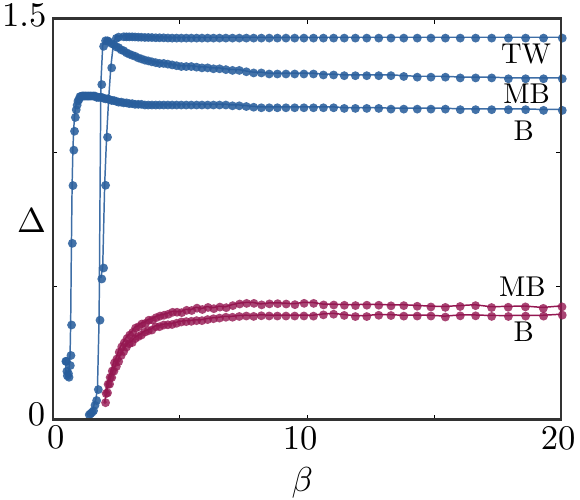}
    \caption{Bifurcation in the control parameter $\beta$. For a
      large range of values, we observe very little change in $\Delta$ as
      $\beta$ is varied. Parameters as in Table~\ref{tab:params} except
      $\kappa=40.0$, $h=0.9$, $p=1.0$. See the main text for full details.}
  \label{fig:continuationBeta}
\end{figure}

We then proceed to verify directly the convergence history of the damped Newton
solver. In Figure~\ref{fig:newtonConvergence}(a), we use a damping factor $0.5$ and
show the residual of the problem as a function of the number of iterations, showing
that the method converges quickly to a solution. At first sight, it is surprising
that the achievable tolerance of the problem does not change when the number of
realisations increases. A second experiment, however, reported in
Figure~\ref{fig:newtonConvergence}(b), shows that this behaviour is caused by
the low system size: when we increase $N$ from $2^7$ to $2^9$, the achievable
tolerance decreases by one order of magnitude.

\subsection{Numerical Bifurcation Analysis}
Gong and Robinson~\cite{Gong2012aa}, and Qi and Gong~\cite{Qi2015aa} found
wandering spots and propagating ensembles using direct numerical simulations on
the plane. Here, we perform a numerical bifurcation analysis with various control
parameters for the structures found in Section~\ref{sec:simulations}
on a one-dimensional domain.

In Figure~\ref{fig:bifDiagComparison}(a), we vary the primary control parameter
$\kappa$, the gain of the convolution term, therefore, we study existence and
stability of the bumps and the travelling pulse when global coupling is varied. This
continuation is performed for a bump, a multiple bump and a travelling pulse in the
continuum deterministic model, using Equations~\eqref{eq:phiBump},
\eqref{eq:multiDeltaEq} and~\eqref{eq:TWEvolution}, respectively.

For sufficiently high $\kappa$, these states coexist and are stable in a large region
of parameter space. We stress that spatially homogeneous mesoscopic states
$J(x) \equiv J_*$, with $0 = J_*$ or $J_* > h$ are also supported by the model, but
are not treated here. Interestingly, the three solution branches are disconnected,
hence the bump analysed in this study does not derive from an instability of the
trivial state. A narrow unstable bump $\Delta \ll 1$ exists for arbitrarily large
$\kappa$ (red branch); as $\kappa$ decreases, the branch stabilises at a saddle-node
bifurcation. At $\kappa\approx 42$, the branch becomes steeper, the maximum of the
bump changes concavity, developing a dimple. On an infinite domain, the branch
displays an asymptote (not shown) as the bump widens indefinitely. On a finite
domain, like the one reported in the figure, there is a maximum achievable width of
the bump, due to boundary effects. The
travelling wave is also initially unstable, but does not stabilise at the saddle node
bifurcation. Instead, the wave becomes stable at $\kappa \approx 33$, confirming the numerical
simulations reported in Figure~\ref{fig:waveInstability}.

In Figure~\ref{fig:bifDiagComparison}(b), we repeat the continuation for the same
parameter values, but on a finite network, using the coarse time-steppers outlined in
Sections~\ref{sec:bumpLift}, \ref{sec:TWpLift}. The numerical procedure returns
results in line with the continuum case, even at the presence of the noise induced by
the finite size. The branches terminate for large $\kappa$ and low $\Delta$: this can
be explained by noting that, if $J(x) \equiv 0$, then the system attains the
trivial resting state $u(x) \equiv 0$ immediately, as no neuron can fire;
on a continuum network, $\Delta$ can be arbitrarily small, hence the branch can be
followed for arbitrarily large $\kappa$; on a discrete network, there is a minimal
value of $\Delta$ that can be represented with a finite grid.

\begin{figure}
  \centering
  \includegraphics{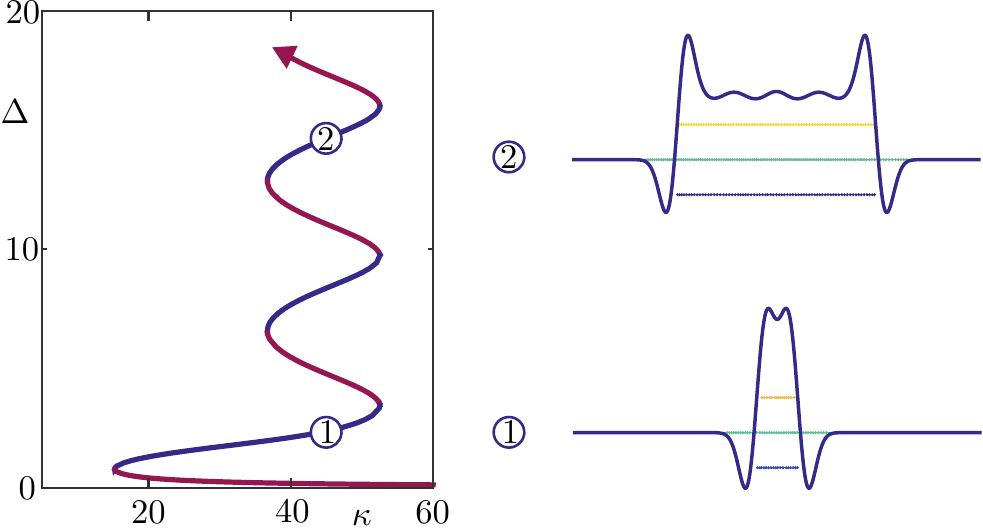}
  \caption{Bifurcation diagram for bumps in a heterogeneous network. To generate
  this figure, we replaced the coupling function with $\widetilde{W}(x,y)=W(x-y) (1+W_0 \cos(y/s))$,
  with $W_0=0.01,\,s=0.5$. We observe the snaking
  phenomenon in the approximate interval $\kappa\in[38,52]$. The branches
  moving upwards and to the right are stable, whereas those moving to the left
  are unstable. The images on the right, obtained via direct simulation, depict
  the solution profiles on the labelled part of the branches. We note the
  similarity of the mesoscopic profiles within the middle of the bump. The
  continuation was performed for the continumm, deterministic model with
  parameters are $\kappa=30$, $h=0.9$.}
  \label{fig:snaking}
\end{figure}

We now consider continuation of solutions in the stochastic model.
In Figure~\ref{fig:continuationP}, we vary the
transition probability, $p$, from the refractory to quiescent state.
In panel (a), we show analytical
results, given by solving \eqref{eq:muBump}-\eqref{eq:hMuBump}, whilst panel
(b) shows results found using the equation-free method.
We find qualitatively similar diagrams in both cases, though we note some
quantitative differences, owing to the finite size of the network and the finiteness
of $\beta$: at the presence of noise, the stationary solutions exist for a wider
region of parameter space (compare the folds in Figure~\ref{fig:continuationP}(a) and
Figure~\ref{fig:continuationP}(b)); a similar situation arises, is also valid for the
travelling wave branches.

The analytical curves of Figure~\ref{fig:continuationP}(a) do not contain any
stability information, which are instead available in the equation-free calculations
of Figure~\ref{fig:continuationP}(b), confirming that bump and multi-bump
destabilise at a saddle-node bifurcation, whereas the travelling wave becomes unstable to
perturbations in the wake, if $p$ is too large. The lower
branch of the travelling wave is present in the analytical results, but not in the
numerical ones, as this branch is not captured by our lifting strategy: when we lift
a travelling wave for very low values of $\Delta$, we have that $J<h$ for all
$x\in \mathbb{S}_N$ and the network attains the trivial state
$u(x) \equiv 0$ in $1$ or $2$ time steps,
thereby the coarse time stepper becomes ineffective, as the
integration time $T$ can not be reduced to $0$.

Gong and co-workers~\cite{Gong2012aa,Qi2015aa} found that refractoriness is a key component to
generating propagating activity in the network. The bifurcation diagram presented
here confirm this, as we recognise 3 regimes: for high $p$ (low refractory time) the
system supports stationary bumps, as the wave is unstable; for intermediate $p$,
travelling and stationary bumps coexist and are stable, while for low $p$ (high
refractory time) the system selects the travelling wave.

In Figure~\ref{fig:continuationBeta}, we perform the same computation now
varying $\beta$, which governs the sensitivity of the transition from quiescence
to spiking. Here, we see that the wave and both bump solutions are stable for
a wide range of $\beta$ values and furthermore, that these states are largely
insensitive to variations in this parameter, implying that the Heaviside limit
is a good approximation for the network in this region of parameter space.

Finally, we apply the framework presented in the previous sections to
study heterogeneous networks. We modulate the synaptic kernel using a harmonic
function, as studied in~\cite{Avitabile2015aa} for a neural field. As
in~\cite{Avitabile2015aa}, the heterogeneity promotes the formation of a
hierarchy of stable coexisting localised bumps, with varying width, arranged in
a classical snaking bifurcation diagram. A detailed study of this bifurcation
structure, while possible, is outside the scope of the present paper.

\section{Discussion}
\label{sec:discussion}
In this article, we have used a combination of analytical and numerical
techniques to study pattern formation in a Markov chain neural network model.
Whilst simple in nature, the model exhibits rich dynamical behaviour, which is
often observed in more realistic neural networks. In
particular, spatio-temporal patterns in the form of bumps have been linked to
working memory~\cite{Funahashi1989aa, Colby1995aa,Goldman-Rakic1995aa}, whilst
travelling waves are thought to be important for plasticity~\cite{Bennett2015aa}
and memory consolidation~\cite{Massimini2004aa,Rasch2013aa}. Overall, our
results reinforce the findings of \cite{Gong2012aa}, namely that refractoriness
is key to generating propagating activity: we have shown analytically and
numerically that
waves are supported by a combination of high gains in the synaptic input
and moderate to long refractory times. For high gains and short refractory
times, the network supports localised, meandering bumps of activity.

The analysis presented here highlights the multiscale nature of the model by
showing how evolution on a microscopic level gives rise to emergent behaviour
at meso- and macroscopic levels. In particular, we established a link
between descriptions of the model at multiple spatial scales: the identified coarse spatiotemporal
patterns have typified and recognisable motifs at the microscopic level, which
we exploit to compute macroscopic patterns and their stability. 

Travelling waves and bumps have almost identical meso- and macroscopic profiles:
if microscopic data were removed from Figure~\ref{fig:exampleBump}(a) and
Figure~\ref{fig:exampleWave}(a), the
profiles and activity sets of these two patterns would be indistinguishable.
We have shown that a disambiguation is however possible if the meso- and
macroscopic descriptions take into account microscopic traits of the patterns:
in the deterministic limit of the system, where mathematical analysis is
possible, the microscopic structure is used in the partition sets of
Propositions~\ref{prop:macroBump} and~\ref{prop:TW}; in the stochastic setting
with Heaviside firing rates and infinite number of neurons, the microscopic
structure is reflected in the approximate probability mass functions appearing
in Section~\ref{sec:approxMu}; in the full stochastic finite-size setting, where
an analytical description is unavailable, the microscopic structure is
hardwired in the lifting operators of the coarse time-steppers
(Section~\ref{sec:coarseTimeStepper}).

An essential ingredient in our analysis is the dependence of the Markov chain
transition probability matrix upon the global activity of the network, via the
firing rate function $f$. Since this hypothesis is used to construct rate models
as Markov processes~\cite{Bressloff:2010jc}, our lifting strategy could be used in
equation-free schemes for more general large-dimensional neural networks. An
apparent limitation of the procedure presented here is its inability to lift
strongly unstable patterns with low activity, as pointed out in
Section~\ref{sec:numerics}. This limitation, however, seems to be specific to
the model studied here: when $\Delta \to 0$, bumps destabilise with transients
that are too short to be captured by the coarse time-stepper.

A possible remedy would be to represent the pattern via a low-dimensional,
spatially-extended, spectral discretisation of the mesoscopic profile
(see~\cite{Laing:2006bg}), which would allow us to represent the synaptic activity
below the threshold $h$. This would lead to a larger-dimensional coarse system,
in which noise would pollute the Jacobian-vector evaluation and the convergence
of the Newton method. Variance-reduction techniques~\cite{Rousset:2013fq} have
been recently proposed for equation-free methods in the context of agent-based
models~\cite{Avitabile2014aa}, and we aim to adapt them to large neural networks
in subsequent publications.

\section*{Acknowledgments} We are grateful to Joel Feinstein, Gabriel Lord and
Wilhelm Stannat for helpful discussions and comments on a preliminary draft of
the paper. Kyle Wedgwood was generously supported by the Wellcome Trust
Institutional Strategic Support Award (WT105618MA)

\bibliographystyle{plain}
\bibliography{references}

\end{document}